\documentclass[12pt,dvips]{amsart}
\usepackage{euler, amssymb, epsfig, hyperref}
\usepackage[pdftex]{color}

\include{defn}           

\newtheorem{Theorem}{Theorem} 
\newtheorem{Proposition}{Proposition} 
\newtheorem{Lemma}{Lemma}

\newtheorem{Corollary}{Corollary}
\newtheorem*{Corollary*}{Corollary}
\newtheorem{Definition}{Definition} 
\newtheorem*{Theorem*}{Theorem}
\theoremstyle{remark} 
\newtheorem{rem}{Remark}

\newcommand\junk[1]{}

\newcommand\GG{\Gamma}
\newcommand\BV{{\bf V}}

\newcommand\Bdd{{\bf \partial}}

\newcommand\BE{{\bf E}}
\newcommand\dd{{\partial}}
\newcommand\GD{{\Delta}}
\newcommand\smin{{\smallsetminus}}
\newcommand\konj{\overline}
\newcommand\A{{\mathbb A}}
\newcommand\Gg{\gamma}
\newcommand\BS{{\bf S}}
\newcommand\BG{{\bf G}}

\newcommand\cT{{\mathcal T}}

\newcommand\Z{{\mathbb Z}}
\renewcommand\L{{\mathbb L}}
\newcommand\cV{{\mathcal V}}

\renewcommand\P{{\mathbb P}}
\newcommand\C{{\mathbb C}}
\newcommand\R{{\mathbb R}}

\newcommand\Go{\omega}

\newcommand\cL{{\mathcal L}}
\newcommand\cS{{\mathcal S}}
\newcommand\cG{{\mathcal G}}
\newcommand\cF{{\mathcal F}}
\newcommand\cN{{\mathcal N}}
\newcommand{\bG}{{\mathbb G}}
\newcommand{\bT}{{\mathbb T}}
\newcommand\cM{{\mathcal M}}

\theoremstyle{plain}

\newcommand\codim{{\rm co}\!\dim}

\numberwithin{equation}{section}

\begin{document}
\pagestyle{plain}

\title{Feynman integrals and motives of configuration spaces}
\author[Ceyhan]{\"Ozg\"ur Ceyhan}
\author[Marcolli]{Matilde Marcolli}
\date{December 25, 2010}
\address{\"O.~Ceyhan: Korteweg-de Vries Institute for Mathematics, University of Amsterdam
P. O. Box 94248, 1090 GE Amsterdam, Netherlands} 
\email{o.ceyhan@uva.nl}
\address{M.~Marcolli: Mathematics Department, California Institute of Technology \\
1200 E. California Blvd.f \\ Pasadena, CA 91125, USA} 
\email{matilde@caltech.edu}

\date{\today}

\begin{abstract}
We formulate the problem of renormalization of Feynman integrals and its relation to
periods of motives in configuration space instead of momentum space.
The algebro-geometric setting is provided by the wonderful compactifications $\overline{Conf}_\Gamma(X)$ of arrangements of subvarieties associated to the subgraphs of a 
Feynman graph $\Gamma$, with $X$ a (quasi)projective variety.
The motive and the class in the Grothendieck ring are computed explicitly
for these wonderful compactifications, in terms of the motive of $X$ and the combinatorics of the
Feynman graph, using recent results of Li Li. 
The pullback to the wonderful compactification of the
form defined by the unrenormalized Feynman amplitude has singularities
along a hypersurface, whose real locus is contained in the exceptional divisors of
the iterated blowup that gives the wonderful compactification. A regularization
of the Feynman integrals can be obtained by modifying the cycle of
integration, by replacing the divergent locus with a Leray coboundary. The
ambiguities are then defined by Poincar\'e residues. While these
residues give mixed Tate periods associated to the cohomology of the
exceptional divisors and their intersections, the regularized integrals
give rise to periods of the hypersurface complement in the wonderful
compactification, which can be motivically more complicated.
\end{abstract}

\maketitle

\section{Introduction}

In recent years a lot of attention has been devoted to motivic aspects of perturbative quantum field
theory, aimed at providing an interpretation of Feynman integrals of a (massless, scalar)
quantum field theory and their renormalization in terms of periods of algebraic varieties. If one can control the nature of the motive of the algebraic variety, then one constraints the kind of numbers that can arise as periods. In particular, the original evidence of \cite{BruKre} suggested that multiple zeta values, hence mixed Tate motives would be the typical outcome of these Feynman integral calculations. When computing Feynman integrals in momentum space, the parametric form of Feynman integrals (see \cite{BjDr2}, \cite{ItZu}) expresses the unrenormalized Feynman amplitude as an integral on the complement of a hypersurface defined by the vanishing of the
Kirchhoff polynomial of the graph. The motivic properties of these hypersurfaces have been widely studied. Contrary to an earlier conjecture of Kontsevich (which was verified in \cite{Stem} for graphs with up to 12 edges), these hypersurfaces are not always mixed Tate motives. More precisely,
it was shown by Belkale and Brosnan \cite{BelBr1} that their classes span the (localized) Grothendieck ring of varieties, hence they can be very far from mixed Tate as motives (see, however, \cite{AluMa5} for the case of the Grothendieck ring without localization). More recently, it was proved rigorously by Francis Brown in \cite{Brown} why all the original cases computed in \cite{BruKre} gave rise to periods of mixed Tate motives, while the smallest explicit counterexample to Kontsevich's conjecture was identified by Doryn in \cite{Doryn}, see
also \cite{BroSch}. A reformulation of the original question in terms of the mixed Tate nature of certain relative cohomology groups for divisors in the complement of the determinant hypersurface and intersections of unions of Schubert cells in flag varieties was given in \cite{AluMa3}. For some related aspects of the interaction between Feynman integrals and motives see also \cite{mar}.

It is natural to consider, from a similar motivic perspective, also the dual picture, where the Feynman integrals and the renormalization procedure take place in configuration space, 
instead of working in momentum space. That is the natural setting of Epstein--Glaser renormalization \cite{EpGla}. As was shown in the seminal papers of Axelrod--Singer \cite{AxSing1}, \cite{AxSing2} in the case of Chern--Simons theory, renormalization of Feynman integrals in configuration space is closely related to the algebro-geometric construction
of the Fulton--MacPherson (FM) compactifications of configuration spaces \cite{fm}. In fact, they associate to a Feynman graph a differential geometric version of the FM compactification of the configuration space on the set of vertices of the graph. The result is a real manifold with corners, which is obtained, like the FM compactification, from a series of blowups, and on which the Feynman integrand extends smoothly. In fact, in the FM compactifications of \cite{fm} and of
\cite{AxSing1}, \cite{AxSing2}, one considers the complement of {\em all diagonals}. It was then observed in \cite{kon1}, \cite{kt}, \cite{BT}, that one can consider configuration spaces associated to graphs, where only the diagonals that correspond to edges in the graph are removed. These
also have compactifications, obtained in a similar way. In fact, the resulting compactifications
are a particular case of a wider class of generalizations of the FM compactifications, namely the ``wonderful models" in the sense of De Concini--Procesi, \cite{DecoPro}.  More precisely, the
recent paper of L.~Li \cite{li2} describes a general procedure to construct configuration spaces
and wonderful compactifications associated to certain arrangements of subvarieties. The graph configuration spaces and the compactifications of Kuperberg--Thurston \cite{kt} are shown in
\cite{li2} to be a special case of this general construction. The FM case is also a special
case that corresponds to the complete graph. We show that the configuration spaces
of graphs and their compactifications required for the regularization of Feynman amplitudes are in fact combinatorially the same as those of Kuperberg--Thurston \cite{kt}, 
using the result of Li \cite{li2}.
The use in Epstein--Glaser renormalization of these graph configuration spaces and their compactifications in the 
De Concini--Procesi sense was recently analyzed in the work of  \cite{BerBruKr},  \cite{BerKr},
\cite{Nikolov}. In particular, the recent paper \cite{BerBruKr} gives a careful description of several geometric aspects of Epstein--Glaser renormalization, formulated in terms of the wonderful compactifications of \cite{DecoPro} for graph configuration spaces.  The role of the Connes--Kreimer Hopf algebra in the Epstein--Glaser setting was also discussed in \cite{BerBruKr}, \cite{Nikolov}, while a version of the motivic Galois group incarnation of the renormalization group of \cite{cm}, \cite{cm1} was formulated in the Epstein--Glaser setting in \cite{Cey}. 

Here we give a reformulation of the motivic question in the configuration space setting.
We begin by describing briefly the geometry of our graph configuration spaces $Conf_\Gamma(X)$
and their compactifications $\overline{Conf}_\Gamma(X)$, and showing that they fit in
the general formalism of \cite{li2} and are in fact equivalent to the Kuperberg--Thurston \cite{kt}
compactifications. We then use another recent results of L.~Li, \cite{li}, on the
Chow motives of wonderful compactifications for smooth projective $X$ (as well as a similar 
result for Voevodsky motives in the quasi-projective case) to obtain an explicit formula for 
the class of $\overline{Conf}_\Gamma(X)$ in the Grothendieck ring of varieties. We obtain from
that also an explicit expression for the virtual Hodge polynomial that generalizes to the
$\overline{Conf}_\Gamma(X)$ the known formula of \cite{Cheah}, \cite{Getzler} for the FM case. 

We then concentrate on the residues of divergent Feynman integrals. We show that,
in the log divergent case, by pulling back the form to the wonderful compactification,
one has simple poles along the exceptional divisor of the deepest diagonal. Using 
a regularization obtained by replacing the divergence locus with a Leray coboundary,
we show that the ambiguity in the renormalization is due to a single Poincar\'e residue.
In the case where there are worse than logarithmic divergences, the pullback to the
wonderful compactification has higher order poles along the exceptional divisors and
the Poincar\'e residues in this case correspond to pieces of the Hodge filtration 
on the primitive cohomology. 

These Poincar\'e residues, that measure the ambiguities of the regularization by
Leray coboundaries, determine mixed Tate periods associated to the cohomology
of the exceptional divisors of the iterated blowups and their intersections, while the
regularized integrals give periods of the complement 
$\overline{Conf}_\Gamma(X)\smallsetminus Z_\Gamma$, where the 
hypersurface $Z_\Gamma$ is a quadric determined by the configuration space 
propagators of the graph. This hypersurface complement can be more complicated
motivically, due to the fact that one does not have a good control over the motivic
nature of the intersections of the components of $Z_\Gamma$ away from the real locus.

\section{Configuration spaces and their combinatorial compactifications}

We describe here briefly the geometry of configuration spaces associated to Feynman 
graphs and their wonderful compactifications. Some of what we discuss here can be
traced to the literature on the subject, especially \cite{AxSing2}, \cite{BerKr}, \cite{fm},
\cite{kon1}, \cite{kt}, \cite{li2}. See also the recent extensive treatment in \cite{BerBruKr}.
We focus here on those aspects that we directly need to obtain the explicit formulae for 
the motive, the class in the Grothendieck ring, and the virtual Hodge polynomial.

\subsection{Configuration spaces of graphs}

In the following, by a {\em graph} $\Gamma$ we always mean a finite graph. We use
the notation $\BV_\Gamma$ for the set of vertices of $\Gamma$ and $\BE_\Gamma$
for the set of edges, and we write $\Bdd_{\GG}: \BE_{\GG} \to S^2(\BV_{\GG})$ for the
boundary map that assigns to an edge its endpoints. (We consider here the graph as
un-oriented, hence the endpoints are defined up to ordering, in the symmetric product
$S^2(\BV_{\GG})$.) A looping edge is an edge for
which the two endpoints coincide and multiple edges are edges between the same pair
of endpoints. We assume that all our graph have no multiple edges and no looping edges,
see Remark \ref{multedgerem} below.

For a subgraph $\gamma \subseteq \Gamma$ we write $\Gamma//\gamma$ to
indicate the graph obtained from $\Gamma$ by shrinking each connected
component of $\gamma$ to a single (different) vertex, and then replacing 
each set of multiple edges with a new single edge. Similarly, we denote by $\Gamma/\gamma$
the quotient where all of $\gamma$ is identified to the same vertex and then each set of 
multiple edges is identified to one single edge. 

Notice that, even though we require the original graph to be free of multiple edges
and looping edges, the quotient graphs can in general have both, hence the reason
why we identify multiple edges in the quotients $\Gamma//\gamma$ and $\Gamma/\gamma$. 
Replacing multiple edges by simple edges does not affect anything in the construction,
see Remark \ref{multedgerem}.
The problem of looping edges in the quotients does not arise, as long as we consider
only induced subgraphs, in the sense of Definition \ref{SGdef}. 

\begin{Definition}\label{CGammaDef}
Let $X$ be a smooth quasi-projective variety and let $\Gamma$ be a graph. 
The configuration space $Conf_\GG(X)$ of $\GG$  in $X$ is the complement 
in the cartesian product $X^{\BV_\GG}=\{(x_v \mid v \in \BV_\GG)\}$ 
of the diagonals associated to the edges of $\GG$, namely
\begin{equation} \label{eqn_open}
Conf_\GG(X) \cong X^{\BV_\GG} \smin \bigcup_{e \in \BE_\GG} \GD_e,
\end{equation}
with 
\begin{equation}\label{Deltae}
\GD_e \cong \{(x_v \mid v \in \BV_\GG) \mid x_{v_1} = x_{v_2} \ {\text for} \ \Bdd_\GG(e) =  \{v_1,v_2\} \}.
\end{equation}
\end{Definition}

\begin{rem}\label{degenmapsG}
By identifying the product $X^{\BV_\GG}$ with the set of all maps
$f: \BV_\GG \to X$, one sees that the configuration space $Conf_\GG(X)$
consist of those maps that are ``non-degenerate along the edges of
$\GG$", that is, such that $f(v)\neq f(v')$ whenever $v$ and $v'$ are connected
by an edge in $\GG$. 
Notice that one can also consider the configuration space of all
non-degenerate maps $f: \BV_\GG \to X$, that is, all maps such that $f(v)\neq
f(v')$ whenever $v \neq v'$. This would correspond to removing
all the diagonals $x_v = x_{v'}$ from $X^{\BV_\GG}$, regardless of whether
the vertices $v$ and $v'$ are connected by an edge in $\GG$ or not. This would
correspond to the configuration space of Definition \ref{CGammaDef} above,
but for the {\em complete graph} with the same set of vertices $V_\GG$ as $\GG$.
\end{rem}

\begin{rem}\label{multedgerem}
Note that the definition of configurations does not detect multiple edges 
in the graph. In fact, in essence the notion of degeneration that defines the
diagonals \eqref{Deltae} is based on collisions of points and not on contracting 
the edges connecting them.  This is why we can assume, to begin with, that the graphs 
we consider have no multiple edges.
On the other hand,  the definition of configuration space is void
in the presence of  looping edges. In fact, a looping edge only
gives the trivial equivalence relation $x_v=x_v$, so that the diagonal
$\Delta_e$ associated to a looping edge is the whole space
$X^{\BV_\Gamma}$, and the complement $X^{\BV_\Gamma}\smallsetminus \Delta_e=\emptyset$.
To avoid this degenerate case, we also assume that graphs have no looping edges.
As observed above, the quotients by induced subgraphs (in the sense of Definition \ref{SGdef} below) will then also have no looping edges.
\end{rem}

\subsubsection{Subgraphs and corresponding diagonals}

We now consider diagonals associated not only to edges of a graph $\Gamma$,
but to certain classes of subgraphs $\gamma \subseteq \Gamma$.

\begin{Definition}\label{SGdef}
A subgraph $\gamma \subseteq \Gamma$ is called an {\em induced subgraph} if
two vertices $v,v'\in \BV_\gamma$ are connected by an edge $e\in \BE_\gamma$
if and only if they are connected by an edge $e\in \BE_\Gamma$, that is, $\gamma$
has all edges of $\Gamma$ on the same set of vertices. Let  $\BS\BG(\GG)$ denote
the set of all connected induced subgraphs of $\GG$. Let
\begin{equation}\label{BSGk}
\BS\BG_k (\GG)  = \{\Gg \in  \BS\BG(\GG) \mid  |\BV_\Gg|=k \},
\end{equation}
be the subset $\BS\BG_k (\GG) \subseteq \BS\BG(\GG)$ of all the connected
induced subgraphs on $k$ vertices. Then $\BS\BG(\GG)$ is a
disjoint union $\BS\BG(\GG) =\cup_{k=1}^{|\BV_\GG|} \BS\BG_k (\GG)$,
where $\BS\BG_{|\BV_\GG|} (\GG)=\{ \GG \}$.
Also let $\widehat{\BS\BG}(\Gamma)$ denote the set of all subgraphs $\gamma$ that
are unions of disjoint connected induced subgraphs. One similarly has subsets
$\widehat{\BS\BG}_k(\Gamma)\subseteq \widehat{\BS\BG}(\Gamma)$ of subgraphs with a
fixed number of vertices. 
\end{Definition}

We now consider diagonals associated to the induced subgraphs in the
following way.

\begin{Definition}\label{DgDef}
For each induced subgraph $\Gg$, the corresponding diagonal is
\begin{equation}\label{oldDgeq}
\GD_\Gg = \{(x_{v_1},\cdots,x_{v_n})\in X^{\BV_\GG}
\mid   x_{v_i} = x_{v_j} \ {\text for \ all}\  v_i,v_j \in \BV_\Gg \}
\end{equation}
while the (poly)diagonal  is
\begin{equation}\label{Dgeq}
\hat\GD_\Gg = \{(x_{v_1},\cdots,x_{v_n})\in X^{\BV_\GG}
\mid   x_{v} = x_{v'} \ {\text for } \
\{ v,v' \} =\Bdd(e),  e\in \BE_\gamma \}.
\end{equation}
\end{Definition}

We then have the following simple property.

\begin{Lemma}\label{diagslem1}
For an arbitrary graph $\Gamma$ and an induced subgraph $\gamma$, the
diagonal $\GD_\Gg$ is isomorphic to  $X^{\BV_{\Gamma/\Gg}}$, while the
(poly)diagonal $\hat\GD_\Gg$
is isomorphic to $X^{\BV_{\Gamma//\Gg}}$. When the graph $\gamma$ is
connected, then $\Delta_\gamma=\hat\Delta_\gamma$.
\end{Lemma}

\proof
In the case where $\gamma$ is not necessarily connected, 
an element $(x_v)\in \hat\Delta_\gamma$
has $x_v=x_{v'}$ for all $v,v'\in \BV_\gamma$ that belong to the same connected
component of $\gamma$. Thus, one can identify
$\hat\Delta_\gamma$ with $X^{\BV_{\Gamma//\gamma}}$. The space
$\Delta_\gamma$ sits as a diagonal in $\hat\Delta_\gamma$ where the values $x_v$
assigned to vertices in the different connected components all agree. It can be
identified with $X^{\BV_{\Gamma/\Gg}}$ where all of $\gamma$ is reduced to just one vertex.

When the graph $\gamma$ is connected, $\Gamma//\gamma =\Gamma/\gamma$ 
is the  graph where all of $\gamma$ is identified to a single vertex. 
One then has an isomorphism between the subspace $\Delta_\gamma$ of
$X^{\BV_\Gamma}$ and the space $X^{\BV_{\Gamma//\gamma}}$.  
\endproof

One can see easily how, in the case of subgraphs that are not connected, the intersection of the diagonals $\Delta_\gamma$ does not behave as nicely as the intersection of the $\hat\Delta_\gamma$. For example, 
let $\gamma\subseteq \Gamma$ be an induced subgraph with two connected components $\gamma= \gamma_1 \cup \gamma_2$. Then $\hat\Delta_\gamma = \hat\Delta_{\gamma_1} \cap \hat\Delta_{\gamma_2}$, while
$\Delta_\gamma \subsetneq  \Delta_{\gamma_1} \cap \Delta_{\gamma_2}$.
This observation follows directly from the previous lemma, using
$\hat\Delta_\gamma = X^{\BV_{\Gamma//\gamma}}$ and $\Delta_\gamma = 
X^{\BV_{\Gamma/\Gg}}$ and the fact that, for the connected graphs $\gamma_i$,
one has $\hat\Delta_{\gamma_i}=\Delta_{\gamma_i}$. 
These have dimensions
\begin{equation}\label{dimDeltas}
\begin{array}{rll}
\dim \Delta_\gamma & =  \dim X^{\BV_{\Gamma/\gamma}}
& =  \dim(X) (|\BV_\Gamma| - |\BV_\gamma| +1), \\  \dim \hat\Delta_\gamma & = \dim
X^{\BV_{\Gamma//\gamma}} & =  \dim(X) (|\BV_\Gamma| - |\BV_\gamma| +b_0(\gamma)).
\end{array}
\end{equation}

\begin{Lemma}\label{disjunLem}
For any graph $\Gamma$, if $\gamma_1$ and $\gamma_2$ are disjoint 
induced subgraphs,
with $\gamma= \gamma_1 \cup \gamma_2$ their disjoint union, then $\hat \Delta_{\gamma_1}$ and $\hat \Delta_{\gamma_2}$ intersect transversely with
$\hat\Delta_\gamma =\hat \Delta_{\gamma_1} \cap \hat \Delta_{\gamma_2}$.

For any graph $\Gamma$, if $\gamma_1$ and $\gamma_2$ are induced subgraphs
which intersect at a single vertex $\gamma_1\cap \gamma_2 =\{ v\}$, then the diagonals
$\hat \Delta_{\gamma_1}$ and $\hat \Delta_{\gamma_2}$ also intersect transversely
with $\hat\Delta_\gamma =\hat \Delta_{\gamma_1} \cap \hat \Delta_{\gamma_2}$, for
$\gamma = \gamma_1 \cup \gamma_2$. 
\end{Lemma}

\begin{rem}\label{induceunion}
The union of induced subgraphs is not in 
general an induced subgraph, as one can see by taking two sides of a triangle, 
or, for disjoint unions, the opposite sides of a square. 
However, one can still define $\hat\Delta_\gamma$ as in  \eqref{Dgeq}.
\end{rem}

\proof {\em (Lemma \ref{disjunLem})} In the case of  disjoint induced subgraphs, the inclusion
$\hat\Delta_\gamma \subseteq \hat \Delta_{\gamma_1} \cap \hat \Delta_{\gamma_2}$
is certainly satisfied. Since the two graphs are disjoint, 
$\Gamma //  \gamma = (\Gamma //  \gamma_1) // \gamma_2=
 (\Gamma //  \gamma_2) // \gamma_1$, and the reverse inclusion also holds.
The dimension counting \eqref{dimDeltas} gives
$\dim(X^{\BV_\GG})=\dim(X) |\BV_\GG|=  \dim(X) ((|\BV_\Gamma| - |\BV_{\gamma_1}| +b_0(\gamma_1))+(|\BV_\Gamma| - |\BV_{\gamma_2}| +b_0(\gamma_2))-
(|\BV_\Gamma| - |\BV_\gamma| +b_0(\gamma)) = \dim(\hat\Delta_{\gamma_1})+\dim(\hat\Delta_{\gamma_2})-\dim(\hat\Delta_\gamma)$.

In the second case, let $\gamma_{1,i}$ and $\gamma_{2,j}$ be the connected
components of $\gamma_1$ and $\gamma_2$, respectively, numbered so that
$\gamma_{1,0}$ and $\gamma_{2,0}$ are the components that contain the vertex $v$.
Then a point $(x_v)$ in the intersection $\hat \Delta_{\gamma_1} \cap \hat \Delta_{\gamma_2}$
satisfies $x_v=x_{1,i}$ for all $v\in \BV_{\gamma_{1,i}}$ and $x_v=x_{2,j}$ for all
$v\in \BV_{\gamma_{2,j}}$ and with $x_{1,0}=x_{2,0}$, so that set-theoretically
$\hat\Delta_\gamma = \hat \Delta_{\gamma_1} \cap \hat \Delta_{\gamma_2}$. Since
$\gamma_1\cap \gamma_2$ consists of a single vertex, $\hat \Delta_{\gamma_1\cap\gamma_2}
=X^{\BV_\Gamma}$, while $|\BV_{\gamma}|=|\BV_{\gamma_1}|+|\BV_{\gamma_2}|-1$
and $b_0(\gamma)=b_0(\gamma_1)+b_0(\gamma_2)-1$, so that
the same dimension counting as above holds.
\endproof

\begin{rem} 
The proof above essentially needs the condition 
$$\Gamma //  \gamma = (\Gamma //  \gamma_1) // \gamma_2= (\Gamma //  \gamma_2) // \gamma_1$$ 
to be satisfied. The cases, (1) $\gamma_1$ \& $\gamma_2$
are disjoint and (2) $\gamma_1$ \& $\gamma_2$ intersect at a vertex, examined in Lemma
\ref{disjunLem} are the only possible cases. In all other cases, 
$\gamma_1 \not\subset (\Gamma //  \gamma_2)$ and $\gamma_2 \not\subset (\Gamma //  \gamma_1)$ . 
\end{rem}

We have also the following property of the diagonals associated to induced subgraphs.

\begin{Lemma}\label{lemsubgr}
For arbitrary $\Gamma$, if $\gamma_1 \subseteq \gamma_2$ are induced subgraphs,
then $\hat\Delta_{\gamma_2} \subseteq \hat \Delta_{\gamma_1}$ and 
$\Delta_{\gamma_2} \subseteq \Delta_{\gamma_1}$.

If $\gamma_1$ and $\gamma_2$ are connected  
induced subgraphs with  $\gamma_1\cap \gamma_2\neq \emptyset$,
such that neither is a subgraph of the other, and with the property 
that their union $\gamma =\gamma_1\cup \gamma_2$ 
is also an induced connected subgraph, then the diagonals $\Delta_{\gamma_1}$ and 
$\Delta_{\gamma_2}$ intersect transversely along the diagonal $\Delta_\gamma$.

For (not necessarily connected) induced subgraphs $\gamma_1$ and $\gamma_2$,
where neither is a subgraph of the other and such that the number of connected
components satisfies
\begin{equation}\label{b0inclexl}
b_0(\gamma) = b_0(\gamma_1) + b_0(\gamma_2) - b_0(\gamma_1\cap \gamma_2),
\end{equation}
with $\gamma =\gamma_1\cup \gamma_2$, 
the diagonals $\hat\Delta_{\gamma_1}$ and 
$\hat\Delta_{\gamma_2}$ intersect transversely along the diagonal $\hat\Delta_\gamma$.
\end{Lemma}

\proof For  $\gamma_1 \subseteq \gamma_2$, we
have $X^{\BV_{\Gamma//\gamma_1}} \supseteq X^{\BV_{\Gamma//\gamma_2}}$ and
$X^{\BV_{\Gamma/\gamma_1}} \supseteq X^{\BV_{\Gamma/\gamma_2}}$, so the
first property clearly holds.
For the second statement, by the first statement
$\Delta_{\gamma_i}\subseteq \Delta_{\gamma_1\cap \gamma_2}$, and
$\Delta_\gamma \subseteq \Delta_{\gamma_1}\cap \Delta_{\gamma_2}$. 
Since the subgraphs have non-empty intersection, an element $(x_v)\in
\Delta_{\gamma_1}\cap \Delta_{\gamma_2}$ has all coordinates $x_v$ with $v\in \BV_\gamma$
with the same value, hence it is in $\Delta_\gamma$, so that $\Delta_\gamma = \Delta_{\gamma_1}\cap \Delta_{\gamma_2}$. By the counting of
dimensions as in \eqref{dimDeltas} we have
$\dim \Delta_{\gamma_1\cap \gamma_2} = \dim(X) (|\BV_\Gamma| - |\BV_{\gamma_1 \cap \gamma_2}| +1) = \dim(X)((|\BV_\Gamma| - |\BV_{\gamma_1}| +1) + (|\BV_\Gamma| - |\BV_{\gamma_2}| +1) - (|\BV_\Gamma| - |\BV_{\gamma}| +1)= \dim \Delta_{\gamma_1}
+\dim \Delta_{\gamma_2} -\dim \Delta_{\gamma}$. The third case is similar. One always
has $\hat\Delta_\gamma \subseteq \hat\Delta_{\gamma_1}\cap \hat\Delta_{\gamma_2}$,
and one sees in the same way that the reverse inclusion also holds, by breaking the argument
up into connected components and applying the previous result. The dimension counting
then gives $\dim(\hat\Delta_{\gamma_1\cap\gamma_2})= \dim(X) ((|\BV_\Gamma| - |\BV_{\gamma_1}| +b_0(\gamma_1))+(|\BV_\Gamma| - |\BV_{\gamma_2}| +b_0(\gamma_2))-
(|\BV_\Gamma| - |\BV_\gamma| +b_0(\gamma)) = \dim(\hat\Delta_{\gamma_1})+\dim(\hat\Delta_{\gamma_2})-\dim(\hat\Delta_\gamma)$,
where we set $\hat\Delta_{\gamma_1\cap\gamma_2}=X^{\BV_\Gamma}$ if $\gamma_1\cap\gamma_2=\emptyset$.
\endproof

\begin{rem}\label{needinduced}
Notice that we need to restrict to induced subgraphs in order to have transversal
intersections. In fact, consider the example of the triangle graph, with an induced
subgraph given by a single edge and the two adjacent vertices, and a (non-induced) 
subgraph given by the remaining two edges and all three vertices. The diagonals 
associated to these subgraphs do not intersect transversely, since one is contained
in the other. This example clearly does not satisfy \eqref{b0inclexl}.
\end{rem}

\begin{rem}\label{nonconnunion}
The second statement of Lemma \ref{lemsubgr} does not hold if the
union $\gamma =\gamma_1\cup \gamma_2$ is not connected. Take
as $\gamma_1$ and $\gamma_2$ two opposite sides in a hexagon.
Both are connected induced subgraphs and their union is induced,
but not connected. The intersection of the diagonals $\Delta_{\gamma_1}$ and 
$\Delta_{\gamma_2}$ is larger than the diagonal $\Delta_\gamma$.
However, in this case, the third statement of Lemma \ref{disjunLem}
ensures that the problem does not arise when working with the (poly)diagonals
$\hat\Delta_\gamma$, since \eqref{b0inclexl} is satisfied in this case.
\end{rem}

\begin{rem}\label{b0suffnotnec}
The condition \eqref{b0inclexl} is sufficient to guarantee transversal
intersections of the (poly)diagonals $\hat\Delta_\gamma$ but not
necessary, as we see in Proposition \ref{strataConf} below. 
\end{rem}

The outcome of the discussion above is that the best behaved class of (poly)diagonals
to consider in our setting is the collection of subvarieties $\hat\Delta_\gamma$, where 
$\gamma\subseteq \Gamma$ is a union of disjoint (connected) induced subgraphs. 
We see next that, in fact, this class has the right properties required to construct a wonderful 
compactification.

\subsection{The wonderful compactifications of arrangements of subvarieties}\label{wondcompSec}

The recent work of L.~Li \cite{li2} provides a general framework for constructing
wonderful compactifications for configuration spaces associated to arrangements
of subvarieties, which generalize the Fulton--MacPherson compactifications of
\cite{fm}, the wonderful compactifications of De Concini--Procesi \cite{DecoPro}, 
the conical compactifications of MacPherson--Procesi \cite{mpro}, and the 
compactifications of graph configuration spaces considered in Kontsevich \cite{kon1}
and also in Kuperberg--Thurston \cite{kt}. We recall here briefly Li's setting of \cite{li2}
and we describe how it can be used to construct a compactification of the
configuration spaces $Conf_\Gamma(X)$, through a family of (poly)diagonals
$\hat\Delta_\gamma$ as in \eqref{Dgeq}.

In the setting of \cite{li2}, a {\em simple arrangement} $\cS$ of subvarieties of an ambient
smooth quasi-projective variety $Y$ is a finite collection of nonsingular closed 
subvarieties $S_i$ with the
properties that all nonempty intersections of subvarieties in the collection are also 
subvarieties in the collection and that any two $S_i$ and $S_j$ in the collection
intersect cleanly (along a nonsingular subvariety, with the tangent bundle of the
intersection equal to the intersection of the restrictions of the tangent bundles). A {\em building set} $\cG$ for a simple arrangement $\cS$ is a subset of $\cS$ with the property that, for any
$S\in \cS$, the minimal elements in the collection $\{ G\in \cG \, :\, G \supseteq S \}$ intersect
transversely with intersection $S$. These minimal elements are called the $\cG$-{\em factors}
of $S$.

The main result of \cite{li2} shows that, given a building set $\cG$ for a simple arrangement,
one can construct a smooth wonderful compactification $Y_\cG$ of the configuration space
\begin{equation}\label{configLi}
Y\smallsetminus \cup_{G\in \cG} G,
\end{equation}
which has an explicit description as a sequence of iterated blowups. 

\begin{rem}\label{noncompYG}
Notice that, in the case where $Y$ is a smooth projective variety, these are indeed
compactifications, while when one still considers the same construction in the smooth
quasi-projective case, the resulting varieties $Y_\cG$ obtained by this method are
still referred to as ``compactifications" though technically they no longer are.
\end{rem}

Consider now the (poly)diagonals $\hat\Delta_\gamma$ defined as in \eqref{Dgeq}.
First notice that the relation between the case of induced subgraphs and the
case of more general subgraphs is given by the following simple observation.

\begin{Lemma}\label{igamma}
Let $\gamma$ be a connected subgraph of $\Gamma$ and let $\iota(\gamma)$ be
the smallest induced subgraph of $\Gamma$ that contains $\gamma$. Then
$\hat\Delta_\gamma =\hat\Delta_{\iota(\gamma)}$.

Let $\gamma$ be a (not necessarily connected and not necessarily induced)
subgraph of $\Gamma$ and let $\gamma_j$ be the connected components
of $\gamma$. Then $\hat\Delta_\gamma = \cap_j \hat\Delta_{\iota(\gamma_j)}$.
\end{Lemma} 

\proof The graph $\iota(\gamma)$ is obtained by adding to $\gamma$ all edges of
$\Gamma$ between vertices of $\gamma$ that are not already edges of $\gamma$.
Then we can see as in Lemma \ref{diagslem1}
that, for a connected graph, the condition defining $\hat\Delta_\gamma$ is the same as 
that defining $\Delta_\gamma$ in \eqref{oldDgeq},
namely  $x_v=x_{v'}$ for all vertices in $\gamma$. When adding the remaining 
edges of $\Gamma$ between the same set of vertices, this does not add any
new identification, hence one obtains the same diagonal.

In the case where $\gamma$ has several connected components, the condition
defining $\hat\Delta_\gamma$ is that $x_v=x_{v'}=x_j$ for all vertices in a given
connected component $\gamma_j$. This condition again remains unchanged if
one replaces each $\gamma_i$ by $\iota(\gamma_i)$. In fact, the graphs
$\iota(\gamma_i)$ are still mutually disjoint, as the additional edges only connect vertices
within the same component.
\endproof

We then see that the (poly)diagonals $\hat\Delta_\gamma$ of disjoint unions
of connected induced subgraphs form an arrangement of subvarieties.

\begin{Lemma}\label{SsetDeltas}
For a given graph $\Gamma$, let $\widehat{\BS\BG}(\Gamma)$ denote the set of all unions
of pairwise disjoint connected induced subgraphs as in Definition \ref{SGdef}. Then the collection
\begin{equation}\label{cSGamma}
\cS_\Gamma=\{ \hat\Delta_\gamma \,|\, \gamma \in \widehat{\BS\BG}(\Gamma) \}
\end{equation}
is a simple arrangement of (poly)diagonal  subvarieties in $X^{\BV_\Gamma}$.
\end{Lemma}

\proof Let $\gamma_1$ and $\gamma_2$ be unions of disjoint connected
induced subgraphs of $\Gamma$.
If $\gamma_1$ and $\gamma_2$ are themselves disjoint, then the union $\gamma = \gamma_1 \cup \gamma_2$ is also an element in $\widehat{\BS\BG}(\Gamma)$ and
the intersection of the (poly)diagonal $\hat\Delta_{\gamma_1}\cap \hat\Delta_{\gamma_2}
=\hat\Delta_\gamma$ is still an element of $\cS_\Gamma$. 
If $\gamma_1 \cap \gamma_2\neq \emptyset$, then let $\gamma_\alpha$ be the connected components of $\gamma$. By Lemma \ref{igamma}, $\hat\Delta_{\gamma_1}\cap \hat\Delta_{\gamma_2}=\hat\Delta_\gamma=\cap_\alpha \hat\Delta_{\iota(\gamma_\alpha)}$ is also
still an element in the class $\cS_\Gamma$. The intersections are clean as all the $\hat\Delta_\gamma$ are smooth and the criterion of Lemma 5.1
of \cite{li2} characterizing clean intersection as the scheme-theoretic intersection being nonsingular  applies to the case of the (poly)diagonals.
\endproof

We can then identify a $\cG$-building set for the arrangement $\cS_\Gamma$. We first recall some further combinatorial properties of graphs that we need in the following.

A graph $\Gamma$ is 2-vertex-connected (biconnected) if it cannot be disconnected by the
removal of a single vertex. Note that the removal of a vertex in a graph means removal of the vertex along with the open star of edges around it. The graph consisting of a single edge is assumed to be biconnected. (See \cite{AluMa3} for a discussion of k-vertex-connectivity in the context of graph hypersurfaces.) 

Any connected graph $\Gamma$ admits a decomposition into {\em biconnected components}.
Namely, the graph $\Gamma$ is determined by a {\em block tree}, which is a finite
tree whose vertices are decorated by biconnected graphs and whose edges 
correspond to {\em cut-vertices} (or {\em articulation vertices}) of $\Gamma$. The graph
$\Gamma$ is obtained by joining the biconnected graphs at the articulation vertices.

\begin{Lemma}\label{indbiconn}
Let $\Gamma$ be a connected graph and $\gamma \subseteq \Gamma$ an induced
connected subgraph. If $\Gamma_i$ are the biconnected components of $\Gamma$, 
then $\gamma
\cap \Gamma_i$ is either empty or a union of biconnected induced subgraphs $\gamma_{ij}$ attached at cut-vertices, which are the biconnected components of $\gamma$.

If $\gamma \subseteq \Gamma$ is a biconnected subgraph and $\iota(\gamma)$ is
the smallest induced subgraph containing $\gamma$, then $\iota(\gamma)$ is
also biconnected.
\end{Lemma}

\proof The intersection $\gamma \cap \Gamma_i$ is clearly an induced subgraph of $\Gamma_i$.
In fact, each biconnected component $\Gamma_i$ is an induced subgraph of $\Gamma$, and
intersections of induced subgraphs are induced subgraphs. Each $\gamma \cap \Gamma_i$
in turn has a decomposition into biconnected components $\gamma_{ij}$. Each component
$\gamma_{ij}$ is also an induced subgraph. In fact, removing a cut-vertex from an induced subgraph leaves an induced subgraph. 

The second statement follows simply by observing that a cut-vertex for $\iota(\gamma)$ would
necessarily also be a cut-vertex for $\gamma$. In fact, after any additional edge of $\iota(\gamma)$
which is in the open star of the cut-vertex is removed, the further removal of all the other edges in
the open star of the cut-vertex disconnects $\gamma$ so that the vertex is also a cut-vertex for
$\gamma$. Additional edges of $\iota(\gamma)$ not attached to the cut-vertex have
endpoints in the same biconnected component of $\iota(\gamma)$ and 
removing them does not affect the cut vertex, which remains
a cut vertex for $\gamma$, so that, in both cases, $\gamma$ would not be biconnected.
\endproof

We then have the following result. The argument is implicit in Proposition 4.1 of \cite{li2}, but we spell it out here for convenience. 

\begin{Proposition}\label{biconnGset}
For a given graph $\Gamma$, the set
\begin{equation}\label{cGGamma}
\cG_\Gamma =\{ \Delta_\gamma \,|\, \gamma\subseteq \Gamma \ \text{induced, biconnected} \}
\end{equation}
is a $\cG$-building set for the arrangement $\cS_\Gamma$ of \eqref{cSGamma}.
The diagonals associated to the biconnected components of an induced subgraph $\gamma$ are
the $\cG_\Gamma$-factors of $\hat\Delta_\gamma$.
\end{Proposition}

\proof Let $\gamma$ be a union of disjoint induced subgraphs. For each connected
component $\gamma_i$ of $\gamma$ consider the decomposition into its biconnected components $\gamma_{ij}$. These are induced subgraphs, whose diagonals $\Delta_{\gamma_{ij}}=\hat\Delta_{\gamma_{ij}}$ are the minimal elements in the collection $\cG_\Gamma$ containing
the element $\hat\Delta_\gamma$ of $\cS_\Gamma$. We know by the first statement of
Lemma \ref{disjunLem} that $\hat\Delta_\gamma=\cap_i \hat\Delta_{\gamma_i}$ is a 
transverse intersection. Each $\hat\Delta_{\gamma_i}=\cap_j \hat\Delta_{\gamma_{ij}}$ 
is in turn a transverse intersection by the second statement of Lemma \ref{disjunLem}.
\endproof

\begin{rem}\label{ktbuild}
By the second observation in Lemma \ref{indbiconn}, for the elements of the building set $\cG_\Gamma$ we can equivalently drop the requirement that the subgraphs are induced 
and use all biconnected graphs. That gives back the building set used in \cite{li2}, as in 
\cite{kt}.
\end{rem}

We then check that the configuration space \eqref{configLi} associated to this
$\cG$-set is the same as the configuration space $Conf_\Gamma(X)$ of Definition
\ref{CGammaDef}. 

\begin{Lemma}\label{configs}
For a graph $\Gamma$ and a smooth quasi-projective variety $X$, the configuration
space $Conf_\Gamma(X)$ of Definition \ref{CGammaDef} is 
\begin{equation}\label{configs2}
Conf_\Gamma(X) = X^{\BV_\Gamma} \smallsetminus \cup_{\gamma\in \cG_\Gamma}
\Delta_\gamma .
\end{equation}
\end{Lemma}

\proof The subgraphs of $\Gamma$ consisting of a single edge are induced biconnected
subgraphs, so that the inclusion $\cup_{e\in \BE_\Gamma} \Delta_e \subset
\cup_{\gamma\in \cG_\Gamma} \Delta_\gamma$ holds. Conversely, given
an induced biconnected subgraph $\gamma$ of $\Gamma$, $\Delta_\gamma=\hat\Delta_\gamma$
is the set of $(x_v)$ with $x_v=x_{v'}$ for $\{v,v'\}=\Bdd(e)$ for $e\in \BE_\gamma$. Thus,
$\Delta_\gamma \subseteq \Delta_e$ for $e\in \BE_\gamma$. Thus, each
$\Delta_\gamma \subseteq \cup_{e\in \BE_\Gamma} \Delta_e$ and the reverse inclusion
also holds.
\endproof

\subsection{The iterated blowup description}\label{blowupSec}

Then the result of Theorem 1.2 of Li \cite{li2} shows that there is a smooth wonderful compactification 
$\overline{Conf}_\Gamma(X)=X^{\BV_\Gamma}_{\cG_\Gamma}$
of the configuration space $Conf_\Gamma(X)$. This is obtained as the closure of
the image of $Conf_\Gamma(X)$ under the inclusion
\begin{equation}\label{barconfincl}
Conf_\Gamma(X) \hookrightarrow \prod_{\Delta_\gamma \in \cG_\Gamma} Bl_{\Delta_\gamma} X^{\BV_\Gamma}.
\end{equation}
Theorem 1.3 of \cite{li2} shows that this wonderful compactification also has a description as
an iterated sequence of blowups. We recall here briefly how that is obtained, as we will need
it later. 

Recall first that, for a blowup $\pi: Bl_Z(Y) \to Y$ of a smooth subvariety in a smooth variety,
the {\em dominant transform} of an irreducible subvariety $V$ of $Y$ is the proper transform
if $V$ is not contained in $Z$ and the (scheme-theoretic) inverse image $\pi^{-1}(V)$ if it is (see Definition 2.7 of \cite{li2}).

Enumerate the set $\cG_\Gamma=\{ \gamma_1,\ldots, \gamma_N \}$ in such a way
that, whenever there is an inclusion $\gamma_i \supseteq \gamma_j$, the corresponding indices are ordered with $i\leq j$. Then, for $k=0,\ldots, N$, let $Y^{(0)}=X^{\BV_\Gamma}$ and let $Y^{(k)}$ be the blowup of $Y^{(k-1)}$ along the (iterated) 
dominant transform of $\Delta_{\gamma_k}$. Theorem 1.3 and Proposition 2.13 of \cite{li2}
show that the variety $Y^{(N)}$ obtained through this sequence of iterated blowups is isomorphic to the wonderful compactification $X^{\BV_\Gamma}_{\cG_\Gamma}$,
\begin{equation}\label{blowConf}
Y^{(N)}= \overline{Conf}_\Gamma(X).
\end{equation}

\begin{rem}\label{ktConf}
Proposition \ref{biconnGset} and Lemma \ref{configs} above, together with Proposition 4.1 of \cite{li2}, show that the configurations spaces of graphs and their compactifications we are considering here are combinatorially the same as the Kuperberg--Thurston compactifications of \cite{kt}.
\end{rem}

The result of \cite{li2} also shows to what extent the result of the iterated sequence of
blowups is dependent on the order in which the blowups are performed. 
In particular, this means that, in our case, we can also describe the sequence of blowups
in the following way. For $k=1,\ldots, n=|\BV_\Gamma|$, let $\cG_{k,\Gamma}\subseteq \cG_\Gamma$ be the subcollection $\cG_{k,\Gamma}=\cG_\Gamma \cap \BS\BG_k(\Gamma)$,
where $\BS\BG_k(\Gamma)$, as in \eqref{BSGk}, is the set of connected induced subgraphs 
with $k$ vertices. 

\begin{Proposition}\label{blowupYn}
Let $Y_0=Y^{(0)}=X^{\BV_\Gamma}$. Inductively, let $Y_k$ denote
the blowup of $Y_{k-1}$ along the dominant transform of 
$\bigcup_{\Gg \in \cG_{n-k+1,\Gamma}} \GD_\Gg$. Then $Y_{n-1}$ is
the wonderful compactification
\begin{equation}\label{compConfGX}
\konj{Conf}_\GG(X) = Y_{n-1}.
\end{equation}
\end{Proposition}

\proof This is a special case of the procedure of Theorem 1.3 of \cite{li2}
described above, where we label the elements of $\cG_\Gamma$,  by listing
the subgraphs in $\cG_{n-k+1,\Gamma}$, for $k=1,\ldots, n-1$, by increasing $k$,
with any arbitrary choice of ordering within each of these sets. The last set, for
$k-1$, corresponds to the set $\cG_{2,\Gamma}$ of subgraphs 
consisting of a single edge. 
We have $Y_k=Y_{k-1}$ if there are no biconnected induced subgraphs with
exactly $n-k+1$ vertices. So, if
$\Gamma$ is itself biconnected, then $Y_1$ is the blowup
of $Y_0$ along the deepest diagonal $\GD_\GG$, which parameterizes
the points where the whole $\GG$ is collapsed, and otherwise $Y_1=Y_0$.
In the resulting sequence of blowups
\begin{equation}\label{Yks}
Y_{n-1} \to \cdots \to Y_2 \to Y_1 \to X^{\BV_\GG}
\end{equation}
the order in which the blowups are performed along the (iterated) dominant
transforms of the diagonals $\Delta_\gamma$, for $\gamma \in \cG_{n-k+1,\Gamma}$,
for a fixed $k$, does not matter, for the general reason described in \S 3 of \cite{li2}.
Thus, the intermediate varieties $Y_k$ in the sequence \eqref{Yks} are all well defined. 
\endproof

\begin{rem}\label{notationYk}
The notational difference above between the $Y^{(j)}$ and the $Y_k$
reflects the fact that each $Y_k$ corresponds to several blowups $Y^{(j)}$,
one for each diagonal $\Delta_{\gamma_j}$ with $\gamma_j \in \cG_{n-k+1,\Gamma}$. 
\end{rem}

\begin{rem}\label{fmrem}
The Fulton--MacPherson compactifications \cite{fm} are obtained as the
wonderful compactification $\overline{Conf}_{\Gamma_n}(X)$, for
$\Gamma_n$ the {\em complete graph} on $n$ vertices, where each 
pair of distinct vertices is connected by an edge.
In this case one needs to blow up all the possible diagonals. 
\end{rem}

\begin{rem}\label{ulyrem}
Notice that, in the case of the complete graph $\GG_n$ on $n$ vertices,
besides the usual Fulton--MacPherson compactification, one can also
consider a different sequence of blowups, where one obtains a more
manifestly symmetric construction with actions of the symmetric group
at each stage. These ``polydiagonal compactifications" were
introduced in \cite{Uly}. The difference is that the blowup loci are
in this case not just diagonals but also their intersections. This introduces
a number of additional blowups in the construction and the resulting
spaces map project down onto the Fulton--MacPherson ones.  These
compactifications are also special cases of the general construction
of Li \cite{li2} for a different choice of $\cG$-building set. One can
consider analogs of the compactifications of \cite{Uly} also in the
case of other graphs $\GG$. The difference with respect to the
case we are considering corresponds to the difference
between the minimal and maximal wonderful compactifications in the sense
of \cite{DecoPro}. This has been recently discussed in \cite{BerBruKr}.
\end{rem}

\subsection{Stratification}\label{strataSec}

Theorem 1.2 of \cite{li2} applied to our case also gives an explicit stratification of
$\overline{Conf}_\Gamma(X)$ in terms of divisors. This will also be useful in the
following and we recall it briefly. 

As above, we consider an arrangement $\cS$ of subvarieties and a $\cG$-building set.
Given a {\em flag} $\cF =\{ S_1, \ldots, S_r \}$ of elements in $\cS$, with $S_1\subseteq S_2 \subseteq \cdots \subseteq S_r$, one defines the associated $\cG$-{\em nest}, as in \cite{mpro}, \cite{li2}, as the collection 
\begin{equation}\label{GTnest}
\cG_\cF =\cup_{i=1}^r \{ R_{ij} \, |\, \cG\text{-factors of } S_i \},
\end{equation}
where, as above, the $\cG$-factors of an element $S\in \cS$ are the minimal elements in the collection $\{ R \in \cG \,|\, R \supseteq S \}$.

We consider the arrangement $\cS_\Gamma$ of \eqref{cSGamma}
and the building set $\cG_\Gamma$ of \eqref{cGGamma} associated to a graph $\Gamma$
and a smooth quasi-projective variety $X$. The $\cG_\Gamma$-nests are then described
easily (see \cite{li2}, \S 4.3) using the following simple observation.

\begin{Lemma}\label{biconnunion}
Let $\gamma_1$ and $\gamma_2$ be biconnected subgraphs of $\Gamma$.
If the intersection $\gamma_1 \cap \gamma_2$ contains at least two distinct 
vertices, then the union $\gamma = \gamma_1 \cup \gamma_2$ is biconnected.
\end{Lemma}

\proof If $\gamma$ were not biconnected, then there would be a vertex $v$ in
$\BV_\gamma$ such that $\gamma \smallsetminus \{ v \}$ has more than one
connected component. If the vertex $v$ belongs to either $\gamma_1$ or
$\gamma_2$, but not to the intersection, then the removal of $v$ would
also disconnect either $\gamma_1$ or $\gamma_2$, contrary to the hypothesis
that they are biconnected. Suppose that the vertex $v$ belongs to the intersection 
$\gamma_1\cap \gamma_2$. The two graphs $\gamma_i \smallsetminus \{ v\}$
are both connected since both $\gamma_i$ are biconnected. The graph
$\gamma \smallsetminus \{ v\} = (\gamma_1 \smallsetminus \{ v\}) \cup
(\gamma_2 \smallsetminus \{v \})$ can then be disconnected only if
$(\gamma_1 \smallsetminus \{ v\}) \cap (\gamma_2 \smallsetminus \{v \})=\emptyset$.
\endproof 

This, together with Proposition \ref{biconnGset} gives the characterization of the
$\cG_\Gamma$-nests.

\begin{Definition}\label{rootedtrees}
A forest $\cT$  of nested subgraphs of a given graph $\Gamma$
is a finite collection of rooted trees, where each component
is a finite tree with vertices labelled by connected induced subgraphs $\gamma_i$ 
of $\Gamma$, with the
property that there is an edge (oriented away from the root vertex) from a vertex $\gamma_i$ to a vertex $\gamma_j$ whenever $\gamma_i \supseteq \gamma_j$. We also require that 
graphs $\gamma$ and $\gamma'$ associated to vertices that lie on different branches of
a tree or on different trees of the forest have $\gamma\cap \gamma'=\emptyset$.
\end{Definition}

Recall that $\cS_\Gamma$ is the simple arrangement of all (poly)diagonals $\hat\Delta_\gamma$,
with $\gamma$ in $\widehat{\BS\BG}(\Gamma)$ and that $\cG_\Gamma$ is the corresponding
building set given by the diagonals $\Delta_\gamma$ with $\gamma$ induced biconnected
subgraphs of $\Gamma$.

The flags in $\cS_\Gamma$ and the associated $\cG_\Gamma$-nests are then described
as follows (see \cite{li2}, \S 4.3).

\begin{Proposition}\label{cGnests}
Flags in $\cS_\Gamma$ are in bijective correspondence with forests of 
nested subgraphs. The $\cG_\Gamma$-nests are in bijective correspondence with
the sets of biconnected induced subgraphs with the property that any two subgraphs
$\gamma$ and $\gamma'$ in the set satisfy one of the following:
\begin{enumerate}
\item $\gamma \cap \gamma' = \emptyset$;
\item $\gamma \cap \gamma' =\{ v \}$, a single vertex;
\item $\gamma \subseteq \gamma'$ or $\gamma' \subseteq \gamma$.
\end{enumerate}
\end{Proposition}

\proof A flag $\cF$ in $\cS_\Gamma$ consists of a sequence $\hat\Delta_{\gamma_1}\subseteq
\hat\Delta_{\gamma_2} \subseteq \cdots \subseteq \hat\Delta_{\gamma_r}$ of (poly)diagonals
associated to disjoint unions of induced subgraphs $\gamma_i$.  By definition of the (poly)diagonals and the fact that the subgraphs are disjoint unions of induced subgraphs, we see 
that the subgraphs satisfy $\gamma_r \subseteq \cdots \subseteq \gamma_1$. 
We then construct a forest of nested
subgraphs $\cT_{\cF}$ which has root vertices 
the connected components $\gamma_{rj}$ of the graph $\gamma_r$, and so on, so that the
set of vertices at a distance $r-i$ to the roots are the connected components  
$\gamma_{ij}$ of the graph $\gamma_i$. The tree has an edge from a connected 
component $\gamma_{ij}$ to a connected component $\gamma_{i'j'}$ whenever 
$i' = i+1$ and $\gamma_{i'j'}\subseteq \gamma_{ij}$.  The forest of nested
subgraphs constructed in this way is uniquely determined by the flag $\cF$. Conversely, given a
forest of nested subgraphs $\cT$, we associate to it a flag $\cF_{\cT}$ in $\cS_\Gamma$ by 
setting $S_i =\cap_{\gamma_{ij}} \hat\Delta_{\gamma_{ij}} =\hat\Delta_{\gamma_i}$, where
$\gamma_{ij}$ are all the connected induced subgraphs attached to the vertices of $\cT$ at a distance $r-i$ to the root, and $\gamma_i =\cup_{ij} \gamma_{ij}$. This gives a bijection between
flags and forests of nested subgraphs.

As in \eqref{GTnest}, a $\cG_\Gamma$-nest is then given by the set of 
$\cG_\Gamma$-factors of the elements $\hat\Delta_{\gamma_1}\subseteq
\hat\Delta_{\gamma_2} \subseteq \cdots \subseteq \hat\Delta_{\gamma_r}$ of a flag.
By Lemma \ref{igamma} and Proposition \ref{biconnGset}, the $\cG_\Gamma$-factors
of each $\hat\Delta_{\gamma_i}$ are the $\hat\Delta_{\gamma_{ij}}$ of its biconnected
components $\gamma_{ij}$. These form a set of induced biconnected subgraphs
with the property that any two $\gamma_{ij}$ and $\gamma_{i'j'}$ 
are either nested one inside the other (when $i\neq i'$), or have 
intersection that is either empty or consisting of a single point (when $i=i'$).
\endproof

We then obtain a stratification of the wonderful compactification 
$\overline{Conf}_\Gamma(X)$ as in Theorem 1.2 of \cite{li2}. 

\begin{Proposition}\label{strataConf}
For $\gamma \subseteq \Gamma$ a biconnected induced subgraph, let $E_\gamma$
be the divisor obtained as the iterated dominant transform of $\Delta_\gamma$ 
in the iterated blowup construction \eqref{Yks} of $\overline{Conf}_\Gamma(X)$. Then
\begin{equation}\label{unionstrata}
\overline{Conf}_\Gamma(X) \smallsetminus Conf_\Gamma(X) =
\bigcup_{\Delta_\gamma \in \cG_\Gamma} E_\gamma.
\end{equation}
The divisors $E_\gamma$ have the property that
\begin{equation}\label{intersGnest}
E_{\gamma_1}\cap \cdots \cap E_{\gamma_\ell} \neq \emptyset \Leftrightarrow \{ \gamma_1, \ldots, \gamma_\ell \} \ \text{ is a } \cG_\Gamma\text{-nest}.
\end{equation}
\end{Proposition}

\proof
The statement is a special case of Theorem 1.2 of \cite{li2}, so we do not reproduce the
proof here in detail. For later use, we just comment briefly on the second statement. Notice that,
if $\{ \gamma_1, \ldots, \gamma_\ell \}$ is a $\cG_\Gamma$-nest, then the divisors
$E_{\gamma_i}$ intersect transversely. In fact, they are the $\cG_\Gamma$-factors of elements
$\hat\Delta_\gamma$ of a flag, and by construction $\cG_\Gamma$-factors intersect transversely,
with intersection the given elements of the flag, so that after passing to the (iterated)
dominant transforms (see Proposition 2.8 of \cite{li2}), one finds a nonempty transverse
intersection.

If $\{ \gamma_1, \ldots, \gamma_\ell \}$ is not a $\cG_\Gamma$-nest, then there are two
subgraphs $\gamma_i$ and $\gamma_j$ in this collection, whose intersection $\gamma_i \cap \gamma_j$ contains at least two distinct vertices. In this case, by Lemma 
\ref{biconnunion}, their union $\gamma_{ij}=\gamma_i \cup \gamma_j$ is also a 
biconnected subgraph with a number of vertices larger than that of both $\gamma_i$
and $\gamma_j$. Thus, the (iterated) dominant transform of the diagonal $\Delta_{\gamma_{ij}}$ 
was blown up at an earlier step in the construction of $\overline{Conf}_\Gamma(X)$ as an
iterated blowup. The diagonals $\Delta_{\gamma_i}$ and $\Delta_{\gamma_j}$
intersect along $\Delta_{\gamma_{ij}}$. Even though the graphs $\gamma_i$ and $\gamma_j$
are not part of a $\cG_\Gamma$-nest, their intersection is still transversal. To see this,
notice that the graphs $\gamma_i$ and $\gamma_j$ are connected, and so is $\gamma_{ij}$.
However, the graph $\tilde\gamma_{ij} =\gamma_i \cap \gamma_j$ needs not be connected.
Thus, the number of connected components can violate the relation \eqref{b0inclexl}, and
we cannot deduce transversality directly by the argument of Lemma \ref{lemsubgr}. However,
notice that $\hat\Delta_{\gamma_i}=\Delta_{\gamma_i}$ and $\hat\Delta_{\gamma_j}=\Delta_{\gamma_j}$ are both contained not only in $\hat\Delta_{\tilde\gamma_{ij}}$ as used in
Lemma \ref{lemsubgr} but also in the possibly smaller $\Delta_{\tilde\gamma_{ij}}\subseteq 
\hat\Delta_{\tilde\gamma_{ij}}$. This has dimension
$$ \dim(\Delta_{\tilde\gamma_{ij}}) =\dim(X) (|\BV_\Gamma|-|\BV_{\gamma_1\cap \gamma_2}|
+1), $$
while $\dim (\hat\Delta_{\tilde\gamma_{ij}}) =\dim(X) (|\BV_\Gamma|-|\BV_{\gamma_1\cap \gamma_2}| +b_0(\tilde\gamma_{ij}))$.
Then one has the correct counting of dimensions
$$ \dim(\Delta_{\tilde\gamma_{ij}}) =\dim(\Delta_{\gamma_i}) +
\dim(\Delta_{\gamma_j}) -  \dim(\Delta_{\gamma_{ij}}) . $$
Since $\Delta_{\gamma_i}$ and $\Delta_{\gamma_j}$
intersect transversely along $\Delta_{\gamma_{ij}}$, whose dominant transform 
was already blown up at an earlier stage in the iterated blowup construction, 
the (iterated) dominant transforms $E_{\gamma_i}$ and $E_{\gamma_j}$ no longer
intersect, $E_{\gamma_i}\cap E_{\gamma_j}=\emptyset$.
\endproof

Let $\cN=\{ \gamma_{ij}, \}$ be the $\cG_\Gamma$-nest of a flag $\cF_\cT=\{
\hat\Delta_{\gamma_1}\subseteq \hat\Delta_{\gamma_2} \subseteq \cdots \subseteq \hat\Delta_{\gamma_r} \}$ associated to a forest of nested subgraphs $\cT$. Let $X_\cN$ be the subvariety of
$\overline{Conf}_\Gamma(X)$ defined by the intersection
\begin{equation}\label{XNvar}
X_\cN := \cap_{ij} E_{\gamma_{ij}} 
\end{equation}
of the divisors associated to the graphs in the $\cG_\Gamma$-nest $\cN$. We know by
Proposition \ref{strataConf} that these intersections are nonempty. The forest 
$\cT$ provides a stratification of the varieties $X_\cN$.

\begin{Lemma}\label{intersXT}
Given two varieties $X_{\cN_1}$ and $X_{\cN_2}$ as in \eqref{XNvar}, the intersection
$X_{\cN_1}\cap X_{\cN_2}\neq \emptyset$ if and only if $\cN=\cN_1\cup \cN_2$ is still
a $\cG_\Gamma$-nest. In this case, let $\cF_{\cT_1}$ and $\cF_{\cT_2}$ be flags with
$\cN_1=\cN(\cT_1)$ and $\cN_2=\cN(\cT_2)$ the $\cG_\Gamma$-nests associated to
these flags. In terms of forests of nested subgraphs, $\cN=\cN(\cT)$ corresponds to 
the flag $\cF_\cT$ of the forest $\cT$ given by the union of $\cT_1$ and $\cT_2$.
\end{Lemma}

\proof Given two flags $\cF_{\cT_1}$ and $\cF_{\cT_2}$ associated to forests of nested subgraphs $\cT_1=\{ \gamma_i \}$ and $\cT_2=\{ \gamma'_k \}$, let $\cN_1=\{ \gamma_{ij} \}$ and $\cN_2=\{ \gamma'_{kr} \}$ be the associated $\cG_\Gamma$-nests. Then, by construction, the intersection $X_{\cN_1}\cap X_{\cN_2}$ is nonempty if and only if the union $\cN=\{ \gamma_{ij} \}\cup\{ \gamma_{kr}\}$ is still a $\cG_\Gamma$-nest. Arguing as in Lemma \ref{cGnests}, we can
construct from $\cN$ a forest $\cT$ of nested subgraphs, so that $\cN$ is the $\cG_\Gamma$-nest
of the flag $\cF_{\cT}$. The forest $\cT$ is the union of the forests $\cT_1$ and $\cT_2$. The intersection $\cT_1\cap \cT_2$ which is the largest subforest with common vertices 
(labelled by the same graphs) is counted only once in $\cT$. 
\endproof

We then have the following description of the open stratum.

\begin{Proposition}\label{XNopen}
The open stratum $X_\cN^\circ$ is given by
\begin{equation}\label{XNcirc}
X_\cN^\circ = X_\cN \smallsetminus \bigcup_{\cT': \cT=\cT'/e} X_{\cN(\cT')},
\end{equation}
where the union is over all the forests of nested subgraphs $\cT'$ such that
$\cT$ is obtained from $\cT'$ by contracting a single edge $e$, whose vertices are decorated
by graphs in the following way. The graph $\gamma'$ decorating the vertex of $e$ that is farther
away from the root of the tree containing it is the graph decorating the corresponding 
vertex in $\cT$ and the
graph $\gamma$ decorating the end of $e$ closer to the root is the union of $\gamma'$ and
a single additional $\cG_\Gamma$-factor. The $\cG_\Gamma$-nest $\cN(\cT')$ is the one 
associated to the flag $\cF_{\cT'}$.
\end{Proposition}

\proof Let $\cT'$ be a forest as above. Assume that the edge $e$ of $\cT'$ is attached to
a root vertex and let $\gamma$ be the graph decorating the other end of the edge $e$, and
let $\gamma_{r,1}$ be the component of $\gamma_r$ decorating the vertex of $\cT'$ that is connected in $\cT'$ to the vertex decorated by $\gamma$. Then, if the flag $\cF_\cT$ is given by $\hat\Delta_{\gamma_1}\subseteq \hat\Delta_{\gamma_2} \subseteq \cdots \subseteq \hat\Delta_{\gamma_r}$, the flag $\cF_{\cT'}$ 
is simply given by $\hat\Delta_{\gamma_1}\subseteq \hat\Delta_{\gamma_2} \subseteq \cdots \subseteq \hat\Delta_{\gamma_r} \subseteq \hat\Delta_{\gamma_{r+1}}$, where the graph
$\gamma_{r+1}$ has connected components given by $\gamma$ and all the other components
$\gamma_{rj}$ of $\gamma_r$, for $j\geq 2$. The $\cG_\Gamma$-nest $\cN(\cT')$ is then
given by the same $\cG_\Gamma$-factors for the graphs $\gamma_i$ already in the original
flag $\cF_{\cT}$ together with the $\cG_\Gamma$-factors $\gamma_\alpha$ of the graph 
$\gamma$. If the graph $\gamma$ has a single additional $\cG_\Gamma$-factor $\tilde\gamma$,
in addition to the $\cG_\Gamma$-factors of $\gamma_r$, then the variety $X_{\cN(\cT')}$
is given by the intersection 
\begin{equation}\label{XNTprime}
X_{\cN(\cT')} = \cap_{i,j: i=1,\ldots, r-1} E_{\gamma_{ij}} \cap \cap_{r,j: j\geq 2} E_{\gamma_{rj}} \cap E_{\tilde\gamma},
\end{equation}
where $E_{\tilde\gamma} \subseteq E_{\gamma_{r1}}$. By Proposition \ref{strataConf}
we then see that the top stratum of $X_\cN$ is obtained by subtracting the intersections
with the other $X_{\cN'}$ and, by Lemma \ref{intersXT}, we see that the largest such 
intersections are in fact given by the $X_{\cN(\cT')}$ described above.
\endproof

This gives a decomposition of $\overline{Conf}_\Gamma(X)$ as a disjoint union of
open strata.

\begin{Corollary}\label{openstrConf}
The variety $\overline{Conf}_\Gamma(X)$ is stratified by the pairwise disjoint
subvarieties $X_\cN^\circ$,
\begin{equation}\label{opstr}
\overline{Conf}_\Gamma(X) = Conf_\Gamma(X) \cup \bigcup_{\cN \in \,\, \cG-nests} X_{\cN}^\circ .
\end{equation}
\end{Corollary}

\proof The statement is a direct consequence of Proposition  \ref{strataConf} and Proposition
\ref{XNopen}. 
\endproof

\subsection{Strata and fiber bundles}\label{bundlesSec}

The open strata $X_\cN^\circ$ also have a description as fiber bundles. To see that, we introduce some preliminary notation and terminology.

Let $\A^d$ be the affine space in $d$-dimensions. The group $G_d$ of translations and homotheties acts on $\A^d$ by $\xi \mapsto \lambda \xi +\eta$, for 
$\eta \in \A^d$ and $\lambda \in \bG_m$ a nonzero scalar. For a given graph $\Gamma$, then,
we define the configuration space of $\A^d$ up to translations and homotheties to be the 
quotient
\begin{equation}\label{CGammaAd}
C_\Gamma(\A^d) := Conf_\Gamma(\A^d) /G_d. 
\end{equation}

Let $v_1$ and $v_2$ be two vertices of $\Gamma$ such that there is an edge 
$e\in \BE_\Gamma$ with $\Bdd (e) =\{ v_1, v_2 \}$.
The configuration space $C_\Gamma(\A^d)$ of \eqref{CGammaAd} can then
be identified (non-canonically) with
\begin{equation}\label{x1x2fix}
C_\Gamma(\A^d) \simeq \{ (x_v)_{v\in \BV_\Gamma} \in Conf_\Gamma(\A^d)  \, |\, 
x_{v_1}=(0,\dots,0) ,  \,\, x_{v_2} =(1,0,\dots,0)\},
\end{equation}
since fixing these coordinates suffices to determine a section of the $G_d$ action.

\begin{Lemma}\label{CGAdLem}
The configuration space $C_\Gamma(\A^d)$ has a nonsingular wonderful ``compactification" 
$\overline{C}_\Gamma(\A^d)$ obtained as in Proposition \ref{blowupYn}.
\end{Lemma}

\proof One can construct, as in Proposition \ref{blowupYn}, the space 
$\overline{Conf}_\Gamma(\A^d)$, as an iterated blowup of $\A^{d |\BV_\Gamma|}$.
(Notice that, technically, this is not a compactification in this case.) To obtain $\overline{C}_\Gamma(\A^d)$ we need to check the compatibility of the construction with the action
of the group $G_d$ of translations and homotheties. One can do this by choosing a section
as in \eqref{x1x2fix} and realize in this way the configuration space $C_\Gamma(\A^d)$
(non-canonically) as a subspace of $Conf_\Gamma(\A^d)$. Then the space $\overline{C}_\Gamma(\A^d)$ is the restriction to this subspace of the ``compactification" $\overline{Conf}_\Gamma(\A^d)$.  This can be seen also by considering the original definition of
$\overline{Conf}_\Gamma(\A^d)$, not in terms of iterated blowups but as the 
closure of $Conf_\Gamma(\A^d)$ inside the space
$\prod Bl_{\Delta_\gamma \in \cG_\Gamma} (\A^{d |\BV_\Gamma|})$.
Then when we look only at those configurations as in \eqref{x1x2fix}, we allow
only those degenerations that do not collapse $x_{v_1}$ and $x_{v_2}$ together and we
obtain the closure of the subspace identified by this choice of section with $C_\Gamma(\A^d)$
inside the same product space. Another way to see this, which does not require choosing
a section of the $G_d$-action as in \eqref{x1x2fix}, is by considering the configuration space
$C_\Gamma(\A^d)$ as a subspace of the quotient $X^{\BV_\Gamma}/G_d$ by the action of
translations and homotheties. One then applies the same iterated blowup construction
described before, but with the $\cG_\Gamma$ building set given by the images of the
diagonals $\Delta_\gamma$ in the quotient $X^{\BV_\Gamma}/G_d$.
\endproof

Now consider again the description of the wonderful compactifications $\overline{Conf}_\Gamma(X)$ as the closure \eqref{barconfincl} of $Conf_\Gamma(X)$ in the space 
\begin{equation}\label{prodBlgamma}
 \prod_{\Delta_\gamma \in \cG_\Gamma} Bl_{\Delta_\gamma} X^{\BV_\Gamma}. 
\end{equation}
By Lemma \ref{configs} we know that we can write $Conf_\Gamma(X)$ as the complement
\eqref{configs} of the diagonals $\Delta_\gamma$ with $\gamma\in \cG_\Gamma$. Then,
in order to describe the strata of the closure of $Conf_\Gamma(X)$ in \eqref{prodBlgamma},
we need to describe the datum over a point where different coordinates $x_v$ and $x_{v'}$, with $v,v'\in \BV_\gamma$ for some graph $\gamma \in \cG_\Gamma$, collide to the same value $x\in X$. Arguing as in \S 1 of \cite{fm}, we see that this datum consists of a collection $(\xi_v)$ of vectors in the tangent space $\bT_x =T_x(X)$, parameterized by the vertices 
$v\in \BV_\gamma$, such that not all coordinates $\xi_v$ are equal.
These data maintain the infinitesimal information on the tangent directions to the points $x_v$ when they collide. These data are defined only up to translation and homotheties, so that, in fact, they define a point $\xi=(\xi_v)_{v\in \BV_\gamma}$ in the projective space 
\begin{equation}\label{projspaceTx}
\xi \in \P(\bT_x^{\BV_\gamma}/\bT_x).
\end{equation}
Such a point $\xi$ is called a {\em screen} for $\gamma$, in the terminology of \cite{fm}.

We introduce the following notation that will be useful later.

\begin{Definition}\label{defGdeltaG}
Given a graph $\Gamma$ and a forest $\cT$ of nested subgraphs 
as in Definition \ref{rootedtrees}, with $\cG$-nest $\cN=\cN(\cT)$.
We denote by $\Gamma/\delta_\cN(\Gamma)$ the graph obtained
as the quotient 
\begin{equation}\label{GdeltaG}
\Gamma/\delta_{\cN}(\Gamma):=\Gamma // (\gamma_1 \cup \cdots \cup \gamma_r),
\end{equation}
for $\cN =\{ \gamma_1, \ldots, \gamma_r \}$. Similarly, for $\gamma$ an induced biconnected subgraph, $\gamma \in \cG_\Gamma$, we set
\begin{equation}\label{gdeltaG}
\gamma/\delta_{\cN}(\gamma):=\gamma // (\gamma_1 \cup \cdots \cup \gamma_k),
\end{equation}
where $\{ \gamma_1, \ldots, \gamma_k \}$ is the set of $\gamma_i \in \cN$ such that
$\gamma_i \subsetneq \gamma$.
\end{Definition}

We then have the following description of the open strata $X_\cN^\circ$. This is analogous
to what discussed in \S 2 of \cite{fm}. 

\begin{Theorem}\label{bundlestrata}
The open strata $X_\cN^\circ$ are fiber bundles over configuration spaces
$Conf_{\Gamma/\delta_{\cN}(\Gamma)}(X)$, where the fiber $\cF_{\cN}$ is obtained as
a succession of fiber bundles, one for each graph $\gamma$ decorating the
vertices of the forest of nested subgraphs $\cT$, with $\cN=\cN(\cT)$, 
where at each stage the fiber $\cF_\gamma$ is the open subvariety of the
space $\P(\bT_x^{\BV_{\gamma/\delta_{\cN}(\gamma)}}/\bT_x)$ of
screen configurations for the graph~$\gamma/\delta_{\cN}(\gamma)$,
which consist of the distinct labeled $|\BV_{\gamma/\delta_{\cN}(\gamma)}|$-tuples
of points in $\bT_x$ up to translations and homothety.
\end{Theorem}

\proof The stratum $X_\cN$ associated to the $\cG$-nest $\cN=\cN(\cT)$ of a forest
$\cT$ of nested subgraphs is given by the intersection $E_{\gamma_1}\cap \cdots\cap
E_{\gamma_r}$ of the exceptional divisors associated to the graphs in 
the $\cG$-nest. Moreover, we have seen in 
Proposition \ref{XNopen} that the open stratum $X_\cN^\circ$ is obtained by
subtracting from $X_\cN$ all the $X_{\cN(\cT')}$ for all the forests $\cT'$ with
$\cT =\cT'/ e$, with the additional vertex of $\cT'$ decorated by a graph with
a single additional $\cG$-factor with respect to the one in the corresponding
vertex of $\cT$. Under the projection map $\pi: \overline{Conf}_\Gamma(X) \to
X^{\BV_\Gamma}$ of the iterated blowup construction, this corresponds to
subtracting from the intersection 
\begin{equation}\label{intersdiagN}
\Delta_{\gamma_1}\cap \cdots \cap \Delta_{\gamma_r}
\end{equation}
all the further intersections 
\begin{equation}\label{intersplusone}
 \Delta_{\gamma_1}\cap \cdots \cap \hat\Delta_{\gamma_j \cup \gamma_j'} \cap
\cdots \cap \Delta_{\gamma_r}, 
\end{equation}
where $\gamma_j'$ is an additional $\cG$-factor and, by Lemma \ref{lemsubgr}, 
$\hat\Delta_{\gamma_j \cup \gamma_j'}=\Delta_{\gamma_j}\cap \Delta_{\gamma_j'}$.
Upon identifying the diagonal $\Delta_\gamma$ of a biconnected graph $\gamma$
with the space $X^{\BV_{\Gamma/\gamma}}$, we can also identify the intersection
\eqref{intersdiagN} with the space 
$X^{\BV_{\Gamma//(\gamma_1 \cup\cdots \cup \gamma_r)}} = X^{\BV_{\Gamma/\delta_{\cN}(\Gamma)}}$. Now we need to check that subtracting the intersections \eqref{intersplusone}
amounts to considering the subspace $Conf_{\Gamma/\delta_{\cN}(\Gamma)}(X)$ inside
the product $X^{\BV_{\Gamma/\delta_{\cN}(\Gamma)}}$. By arguing as in Lemma \ref{configs},
we see that the complement of the union of the diagonals $\Delta_{\gamma_j'}$ as
above is the same as the complement of the union of the diagonals $\Delta_e$, with
$e$ ranging over the edges of the graph $\Gamma/\delta_{\cN}(\Gamma)$. This proves
that, under the map $\pi: \overline{Conf}_\Gamma(X) \to
X^{\BV_\Gamma}$ of the iterated blowup construction, the image of 
an open stratum $X_\cN^\circ$ can be identified with the configuration spaces 
$Conf_{\Gamma/\delta_{\cN}(\Gamma)}(X)$, for $\cN=\cN(\cT)$. We then check that
$\pi : X_{\cN(\cT)}^\circ \to Conf_{\Gamma/\delta_{\cT}(\Gamma)}(X)$ is a fiber bundle. 
In the iterated blowup construction of $\overline{Conf}_\Gamma(X)$, we have seen
that one progressively blows up diagonals $\Delta_\gamma$ in $\cG_\Gamma$ by
decreasing number of vertices. At each stage, when one of the $\Delta_{\gamma_j}$ is 
blown up, the exceptional divisor is the projectivized normal bundle 
\begin{equation}\label{normbundlegN}
 \P(N_{\Delta_\gamma \subset \cap_{\gamma'\in \cN : \gamma'\subsetneq \gamma} \Delta_{\gamma'} } ), 
\end{equation} 
or the projectivized normal bundle $\P(N_{\Delta_\gamma \subset X^{\BV_\Gamma}})$
when the set $\{ \gamma'\in \cN : \gamma'\subsetneq \gamma \}=\emptyset$. 
This projectivized normal bundle indeed carries the infinitesimal information about the 
degeneration, when points collide along the diagonal $\Delta_\gamma$ and can be
described, as in \S 1 of \cite{fm} in terms of screen configurations. In fact, first observe
that we can identify $\Delta_\gamma \simeq X^{\BV_{\Gamma/\gamma}}$ and similarly we
can identify
$$ \cap_{\gamma'\in \cN : \gamma'\subsetneq \gamma} \Delta_{\gamma'} \simeq
X^{\BV_{\Gamma//(\gamma_1\cup \cdots \cup \gamma_k)}} , $$
with $\{ \gamma_1, \ldots, \gamma_k \}=\{ \gamma'\in \cN :
\gamma' \subsetneq \gamma \}$. Thus, we have
$$ \P(N_{\Delta_\gamma \subset \cap_{\gamma'\in \cN : \gamma'\subsetneq \gamma} \Delta_{\gamma'} }) \simeq \P(T_x(X^{\BV_{\Gamma // (\gamma_1\cup \cdots \cup \gamma_k)}})
/T_x(X^{\BV_{\Gamma/\gamma}})). $$
Then observe that the dimension of this projectivized normal bundle is given by
$$ d (|\BV_\Gamma| - |\BV_{\gamma_1\cup \cdots \cup \gamma_k}| + b_0(\gamma_1\cup \cdots \cup \gamma_k)) - d(|\BV_\Gamma| - |\BV_\gamma| +1), $$
where $d=\dim X$. This is equal to
$$ d(|\BV_\gamma| -  |\BV_{\gamma_1\cup \cdots \cup \gamma_k}| + b_0(\gamma_1\cup \cdots \cup \gamma_k)) -1 ). $$
In fact, we can identify
$$ T_x(X^{\BV_{\Gamma//(\gamma_1\cup \cdots \cup \gamma_k)}})/T_x(X^{\BV_{\Gamma/\gamma}}) \simeq 
\bT_x^{\BV_{\gamma // (\gamma_1\cup \cdots \cup \gamma_k)}} / \bT_x , $$
so that we obtain
$$ \P(N_{\Delta_\gamma \subset \cap_{\gamma'\in \cN : \gamma'\subsetneq \gamma} \Delta_{\gamma'} }) \simeq \P(\bT_x^{\BV_{\gamma/\delta_{\cN}(\gamma)}}/\bT_x), $$ 
with the notation $\gamma/\delta_{\cN}(\gamma)$ as in \eqref{gdeltaG}.
The space $\P(\bT_x^{\BV_{\gamma/\delta_{\cN}(\gamma)}}/\bT_x)$
is exactly the space parameterizing the screen configurations of $\gamma/\delta_{\cN}(\gamma)$ described earlier (see \cite{fm}, \S 1).

Similarly, in the case of the projectivized normal bundle $\P(N_{\Delta_\gamma \subset X^{\BV_\Gamma}})$,  the identification $\Delta_\gamma \simeq X^{\BV_{\Gamma/\gamma}}$, together
with the fact that $|\BV_\Gamma| -|\BV_\gamma|+1=|\BV_{\Gamma/\gamma}|$ gives
at the level of tangent spaces
\begin{equation}\label{identifTxN}
 \bT_x^{\BV_\gamma}/\bT_x \simeq \bT_x^{\BV_\Gamma} /\bT_x^{\BV_{\Gamma/\gamma}} 
\simeq T_x(X^{\BV_\Gamma})/T_x(\Delta_\gamma) \simeq N(\Delta_\gamma \subset X^{\BV_\Gamma}), 
\end{equation}
where $\bT_x=T_x(X)$. Thus, we can identify the projectivization $\P(N(\Delta_\gamma \subset X^{\BV_\Gamma}))$ with the projectivization 
\begin{equation}\label{ProjNeq}
\P(N(\Delta_\gamma \subset X^{\BV_\Gamma}))\simeq \P(\bT_x^{\BV_\gamma}/\bT_x).
\end{equation}
This again is the space parameterizing the screen configurations of $\gamma$.

One can argue as in the proof of Proposition 2.1 of \cite{fm} and identify $\cF_\gamma$ with
the subspace of this space of screen configurations $\P(\bT_x^{\BV_{\gamma/\delta_{\cN}(\gamma)}}/\bT_x)$ (or $\P(\bT_x^{\BV_\gamma}/\bT_x)$) that corresponds to the 
 distinct labeled tuples of points in $\bT_x$ up to translations and homothety.
\endproof

\begin{Corollary}\label{CGammaAdFiber}
In the fiber $\cF_\cN$ of the bundle $\pi: X_\cN^\circ \to Conf_{\Gamma/\delta_{\cN}(\Gamma)}(X)$,
each $\cF_\gamma$ as in Theorem \ref{bundlestrata} is isomorphic to the configuration space 
$C_{\gamma/\delta_{\cN}(\gamma)}(\A^d)$, with $d=\dim(X)$, 
defined as in \eqref{CGammaAd}.
\end{Corollary}

\proof One can identify the tangent space $\bT_x$ with a copy of the affine space $\A^d$. Then the action on $\A^d$ of the group $G_d$ of translations and homothety corresponds to the identifications on $\bT_x$ that describe screen configurations. Thus, for a given graph 
$\gamma$, the projective space $\P(\bT_x^{\BV_{\gamma/\delta_{\cN}(\gamma)}}/\bT_x)$ can be identified with the
quotient $\A^{d |\BV_{\gamma/\delta_{\cN}(\gamma)}|}/G_d$, which contains the configuration space $C_{\gamma/\delta_{\cN}(\gamma)}(\A^d)$. Moreover, the latter describes precisely
those screen configurations that consist of  distinct labeled  $|\BV_{\gamma/\delta_{\cN}(\gamma)}|$-tuples of points in $\bT_x$ up to translations and homothety.
\endproof

\begin{Corollary}\label{XPdcase}
In the case where the variety $X$ is a projective space $\P^d$, the stratum
$X_\cN^\circ$ contains a subspace (non-canonically) isomorphic to $\overline{C}_{\Gamma/\delta_{\cN}(\Gamma)}(\A^d)$.
\end{Corollary}

\proof When $X=\P^d$, inside $Conf_{\Gamma/\delta_{\cN}(\Gamma)}(X)=\pi(X_{\cN}^\circ)$ we have a copy of $Conf_{\Gamma/\delta_{\cN}(\Gamma)}(\A^d) \subset Conf_{\Gamma/\delta_{\cN}(\Gamma)}(\P^d)$. Moreover, by the (non-canonical) choice of a section as in  \eqref{x1x2fix},
we can identify inside this  $Conf_{\Gamma/\delta_{\cN}(\Gamma)}(\A^d)$ a subspace
isomorphic to $C_{\Gamma/\delta_{\cN}(\Gamma)}(\A^d)$. Then the fiber of the map $\pi$
above this space is still given by the screen configurations of the graphs $\gamma$ in the forest
$\cT$ of the $\cG_\Gamma$-nest $\cN$, as in Theorem \ref{bundlestrata}, which, by 
Lemma \ref{CGAdLem}, give the ``compactification" 
$\overline{C}_{\Gamma/\delta_{\cN}(\Gamma)}(\A^d)$. 
\endproof

\section{Motives of configuration spaces}

In the momentum space description, one considers the complement
of a graph hypersurface in a projective space or in a toric variety obtained
as an iterated blowup of projective space \cite{BEK}, \cite{BK}. In an
equivalent reformulation of the momentum space integrals given in
\cite{AluMa3}, one considers a divisor in the complement of a determinant
hypersurface. In all of these cases, one has an ambient space whose
motive can be explicitly described as a mixed Tate motive, while the
hypersurface complement in \cite{BEK}, \cite{BK}, or the intersection 
of the divisor with the hypersurface complement in \cite{AluMa3},
become the loci about which one wants to understand whether they
are motivically mixed Tate or not. 

We consider here first the motive of the ambient space, which
in the configuration treatment is given by the iterated blowup
$\overline{Conf}_\Gamma(X)$ we described in the previous section.
We give an explicit description of the associated motive, based on
the results of L.~Li \cite{li} on motives of wonderful compactifications.

\smallskip

\subsection{Chow motives of configuration spaces}

We state here a first result assuming that $X$ is a smooth {\em projective} variety.
In this case, we can work in the category of Chow motives, and rely directly on the
result of \cite{li}.

The main ingredient that is used in \cite{li} to compute the Chow motive of
the wonderful compactifications is a blowup formula for motives, which
follows from \S 9 of \cite{Man}, and is also proved in Theorem A.2 of \cite{li}.
Namely, if 
$f: \widetilde{Y} \to Y$ is the blow-up of a smooth projective variety $Y$ along
a non-singular closed subvariety $V \subset Y$, then $h(\widetilde{Y})$ is canonically 
isomorphic to 
\begin{equation}\label{Chowblowup}
h(\widetilde{Y}) \cong h(Y) \oplus \bigoplus_{k=1}^{\codim_Y(V)-1} h(V)(k)
\end{equation}
in the category of Chow motives.

Here one uses the standard notation for Chow motives, written as triples $(X,p,m)$
of a variety $X$, a projector $p$, and an integer $m$, where for a smooth
projective varieties $X$ one writes the corresponding motive as $h(X)=(X,id_X,0)$,
and its Tate twists by $h(X)(\ell)=(X,id_X,\ell)$.

We can then obtain the explicit formula for the Chow motive of the
compactifications $\konj{Conf}_\GG(X)$ directly from the main formula 
of \cite{li} for the Chow motive of all the wonderful compactifications. 
We first introduce the following notation. Given a $\cG_\Gamma$-nest $\cN$, and a
biconnected induced subgraph $\gamma$ such that $\cN' =\cN \cup \{ \gamma \}$
is still a $\cG_\Gamma$-nest, we set
\begin{equation}\label{rgamma}
r_\gamma= r_{\gamma,\cN} := \dim (\cap_{\gamma'\in \cN: \gamma'\subset \gamma} \Delta_{\gamma'})-\dim \Delta_\gamma ,
\end{equation}
\begin{equation}\label{MN}
M_\cN:=\{ (\mu_\gamma)_{\Delta_\gamma \in \cG_\Gamma} \,:\, 1\leq \mu_\gamma \leq r_\gamma -1 , \,\, \mu_\gamma \in \Z \}, 
\end{equation}
\begin{equation}\label{normmu}
\| \mu \| := \sum_{\Delta_\gamma \in \cG_\Gamma} \mu_\gamma .
\end{equation}
These agree with the notation used in \cite{li}. We then have the following result.

\begin{Proposition} \label{thm_chow}
Let $X$ be a smooth projective variety of dimension $d$. Then
the Chow motive  of $\konj{Conf}_\GG(X)$ is given by 
\begin{equation}\label{Chowmot}
h(\konj{Conf}_\GG(X)) = h(X^{\BV_\GG}) \oplus \bigoplus_{\cN \in \cG_\Gamma\text{-nests}}\,\,\,\,
\bigoplus_{\mu \in M_\cN} h(X^{\BV_{\Gamma/\delta_\cN(\Gamma)}}) (\|\mu\|).
\end{equation}
\end{Proposition}

\proof The result is a direct consequence of Theorem 3.1 of \cite{li}, which is proved as 
a downward induction on the tower of iterated blowups describing $\konj{Conf}_\GG(X)$,
where at each stage one applies the blowup formula \eqref{Chowblowup}. The only
thing we need to check to match \eqref{Chowmot} to the formula of \cite{li} is that the
motives involved in the second summation are indeed the 
$h(X^{\BV_{\Gamma/\delta_\cN(\Gamma)}})$. 
In Li's formulation, if we denote by $\pi: \konj{Conf}_\GG(X) \to X^{\BV_\Gamma}$ the map
of the iterated blowup, we have in the formula for the Chow motive of a wonderful
compactification $Y_{\cG}$ a sum over $\cG$-nests and, for each $\cG$-nest $\cN$ 
a sum over $\mu$ of $\|\mu\|$-twisted copies of the motive $h(\pi(X_\cN))$, where,
with our notation, $X_\cN=\cap_{\gamma\in \cN} E_\gamma$. To see that $\pi(X_\cN)$
is indeed isomorphic to $X^{\BV_{\Gamma/\delta_\cN(\Gamma)}}$, then notice that
the (poly)diagonal $\pi(\cap_{\gamma\in \cN} E_\gamma)= \cap_{\gamma\in \cN} \Delta_\gamma$ corresponds to identifying the coordinates $x_v$ of all vertices in each connected component
of the graph $\gamma_1 \cup \cdots \cup \gamma_N$, where $\cN=\{ \gamma_1, \ldots, \gamma_N\}$. Thus, we can identify $\pi(\cap_{\gamma\in \cN} E_\gamma)$ with the
space $X^{\BV_{\Gamma/\delta_\cN(\Gamma)}}$, where
$\Gamma/\delta_\cN(\Gamma)= \Gamma //(\gamma_1 \cup \cdots \cup \gamma_N)$.
\endproof

\subsection{Voevodsky motives and the quasi-projective case}

We now extend the result of Proposition \ref{thm_chow} to the case of 
smooth quasi-projective varieties. In this case we can no longer work
with Chow motives, but we need mixed motives in the sense of 
Voevodsky \cite{voe}. The argument, however, is entirely similar,
after one replaces the blowup formula \eqref{Chowblowup} for
Chow motives with the analogous blow-up formula 
for motives in the Voevodsky category. We write here $m(X)$ for the motive in the
Voevodsky category. This corresponds to the notation $M_{gm}$ of \cite{voe}.

Then the blowup formula we need is the one proved in Proposition 3.5.3 of \cite{voe}.
If $f: \widetilde{Y} \to Y$ is the blow-up of a smooth scheme $Y$ along
a smooth closed subscheme $V \subset Y$, then $m(\widetilde{Y})$ is canonically 
isomorphic to 
\begin{equation}\label{Voevblowup}
m(\widetilde{Y}) \cong m(Y) \oplus \bigoplus_{k=1}^{\codim_Y(V)-1} m(V)(k)[2k]
\end{equation}
in the category of Voevodsky's motives.
Here $[-]$ denotes the shift in the triangulated category of mixed motives, while $(-)$ is,
as before, the Tate twist.

As before, we let $Y^{(k)}$ denote the
iterated blowups of $X^{\BV_\Gamma}$ as in \cite{li2}, 
with the wonderful ``compactification" $\konj{Conf}_\GG(X) = Y^{(N)}$,
where $\cG_\Gamma=\{ \gamma_1, \ldots, \gamma_N \}$, ordered, as
before, in such a way that whenever $\gamma_i \supseteq \gamma_j$ the
corresponding indices are ordered by $i\leq j$.

We first introduce the following notation.  For a given $\cG_\Gamma$-nest $\cN$, 
let $X^{(k)}_\cN$ denote the intersection $X^{(k)}_\cN= \cap_{\gamma \in \cN} 
E_\gamma^{(k)}$, where we denote by $E_\gamma^{(k)}$ the iterated dominant
transform in $Y^{(k)}$ of $\Delta_\gamma$. 

\begin{Proposition}\label{thmMTvoe}
Let $X$ be a quasi-projective smooth variety of dimension $d$. 
If $\cN$ is a $\cG_\Gamma$-nest with $\cN \subseteq \{ \gamma_{k+2}, \ldots, \gamma_N \}$,
with the property that $\cN'=\cN \cup \{ \gamma_{k+1} \}$ is also a $\cG_\Gamma$-nest, then
the Voevodsky motives of the subvarieties $X^{(k)}_\cN$ in the
iterated blowup $Y^{(k)}$ of $X^{\BV_\Gamma}$
satisfy the recursion formula
\begin{equation}\label{recXNk}
m(X_\cN^{(k+1)}) =  m(X_\cN^{(k)})  \oplus  \bigoplus_{\ell =1}^{r_{k,\cN}-1} m(X_{\cN'}^{(k)}) (\ell)[2\ell],
\end{equation}
where the codimension $r_{k,\cN}$ is given by
$r_{k,\cN}= \dim( \cap_{\gamma\in \cN: \gamma \subset \gamma_{k+1}} \Delta_\gamma )
-\dim \Delta_{\gamma_{k+1}}$
when $\{ \gamma\in \cN: \gamma \subsetneq \gamma_{k+1} \} \neq \emptyset$ and by
$r_{k,\cN}= d |\BV_\Gamma| - \dim \Delta_{\gamma_{k+1}}$
otherwise.
\end{Proposition}

\proof The proof is entirely similar to the proof of the analogous statement for Chow
motives in the smooth projective case, proved in Lemma 3.3 of \cite{li}, where 
at each step one replaces the use of the blowup formula \eqref{Chowblowup} 
with the formula \eqref{Voevblowup}. 
\endproof

We then have the analog of Proposition \ref{thm_chow}.

\begin{Proposition}\label{propMTvoe}
Let $X$ be a smooth quasi-projective variety. 
The Voevodsky motive $m(\konj{Conf}_\GG(X))$ of the wonderful ``compactification" 
is given by
\begin{equation}\label{voeConfGX}
m(\konj{Conf}_\GG(X)) = m(X^{\BV_\GG}) \oplus \bigoplus_{\cN \in \cG_\Gamma\text{-nests}}
\,\,\,\,  \bigoplus_{\mu \in M_\cN} m(X^{\BV_{\Gamma/\delta_{\cN}(\Gamma)}})
(\|\mu\|)[2\|\mu\|].
\end{equation}
\end{Proposition}

\proof This also follows immediately by the same argument used in the smooth
projective case for Chow motives in the proof of Theorem 3.1 of \cite{li}, where, 
in the downward induction on the levels $k$ of the iterated blowup describing
$\konj{Conf}_\GG(X)$, one replaces at each step with the formula \eqref{recXNk} 
the analogous formula used in \cite{li} for Chow motives. 
\endproof

We obtain then from Propositions \ref{thm_chow} and \ref{propMTvoe}
the following simple corollary.

\begin{Corollary}\label{confMT}
If the motive of the smooth (quasi)projective variety $X$ is mixed Tate,
then the motive of $\konj{Conf}_\GG(X)$ is also mixed Tate, for all graphs
$\Gamma$. In particular, for example, the motives of
$\konj{Conf}_\GG(\P^d)$, $\konj{Conf}_\GG(\A^d)$ 
and $\overline{C}_\Gamma(\A^d)$ are mixed Tate.
\end{Corollary}

\proof This is an immediate consequence of \eqref{Chowmot} and \eqref{voeConfGX},
since the motive of $\konj{Conf}_\GG(X)$ depends upon the motive of $X$ only through
products, Tate twists, sums, and shifts. All these operations preserve the subcategory
of mixed Tate motives.
\endproof

\medskip
\subsection{Classes in the Grothendieck ring}\label{GrclassSec}

The formula for the motive of $\overline{Conf}_\Gamma(X)$ has a 
corresponding formula for a simpler invariant that captures some of
the motivic properties, the class in the Grothendieck ring of varieties 
$K_0(\cV)$. This is generated by isomorphism classes $[X]$
of quasi-projective varieties, with the relations $[X]=[Y]+[X\smallsetminus Y]$
for closed embeddings $Y\subset X$ and with product $[X\times Y]=[X]\cdot [Y]$.

Recall that  an invariant $\chi([X])$ of isomorphism classes of algebraic varieties with values in a 
commutative ring ${\mathcal R}$ is {\em motivic} if it factors through the Grothendieck ring of 
varieties, that is if it satisfies the inclusion--exclusion (or scissor congruence) and product 
relations
\begin{equation}\label{invmot}
\chi([X])=\chi([Y])+\chi([X\smallsetminus Y]) \ \ \  \text{ and } \ \ \ 
\chi([X \times Y]) =\chi([X])\cdot \chi([Y]),
\end{equation}
that is, if it defines a ring homomorphism $\chi: K_0(\cV)\to {\mathcal R}$.
The topological Euler characteristic is a prototype example of such an invariant, 
and for that reason the class $[X]$ in the Grothendieck ring can be regarded as
a {\em universal Euler characteristic}, \cite{Bittner}.

The class in the Grothendieck ring and the motive of a variety are related through
the motivic Euler characteristic. For Chow motives, this was constructed in
\cite{GilSou}, as an invariant $\chi_{mot}((X,p,m))$, satisfying the inclusion--exclusion and
product relation, which associates
to an element $(X,p,m)$ a class in the Grothendieck
ring $K_0(\cM_{Chow})$ of the pseudoabelian category $\cM_{Chow}$ of Chow motives. 
The motivic Euler characteristic of the Chow motive
$h(X)=(X,id_X,0)$ of a smooth projective variety $X$ 
factors through the class $[X]$ in the Grothendieck
ring of varieties $K_0(\cV)[\L^{-1}]$, with the Lefschetz motive inverted,
via a ring homomorphism $\chi: K_0(\cV)[\L^{-1}]\to K_0(\cM_{Chow})$, so
that $\chi_{mot}(h(X)) =\chi([X])$ in $K_0(\cM_{Chow})$.
This motivic Euler characteristic was generalized to the Voevodsky category
of mixed motives in \cite{Bond}. We denote it by $\chi_{mot}(m(X))$.

In the Grothendieck ring, the Lefschetz motive corresponds to $\L = [\A^1]$, the
class of the affine line. The subring $\Z[\L,\L^{-1}]$ of  the Grothedieck ring 
$K_0(\cV)[\L^{-1}]$ is the image of the mixed Tate motives. 

The blowup formulae \eqref{Chowblowup}  and \eqref{Voevblowup}  for motives 
have an analog for the classes in the Grothendieck ring $K_0(\cV)$ of varieties,
namely the Bittner relation \cite{Bittner}. 

These are based on the fact that, for $f: X \to Y$ a locally trivial fibration with fiber $F$, 
the class in the Grothendieck ring of varieties $K_0(\cV)$ satisfies
\begin{equation}\label{trivFibr}
[X] = [Y]\cdot [F].
\end{equation}
This follows directly from the scissor relations defining the Grothendieck ring
and Noetherian induction. This then shows (\cite{Bittner}) that, in the case of a 
blowup $f: \widetilde{Y} \to Y$ of a smooth scheme $Y$ along
a smooth closed subscheme $V \subset Y$, with exceptional divisor $E$, the class
$[\widetilde{Y}]$ in $K_0(\cV)$ satisfies the Bittner relation
\begin{equation}\label{K0blowup}
[\widetilde{Y}] = [Y] - [V] + [E] =[Y] +  [V]  ([\P^{\codim_Y(V)-1}] -1).
\end{equation}
In fact, it is shown in \cite{Bittner} that this relation can be used as a
replacement for the inclusion-exclusion relation $[X]=[Y]+[X\smallsetminus Y]$
for closed embedding, in the construction of the Grothendieck ring of
varieties.

We write this equivalently in the following form, which is more
similar to the form of \eqref{Chowblowup}  and \eqref{Voevblowup}.

\begin{Lemma}\label{K0blowup}
The class $[\widetilde{Y}]$ of a blowup $f: \widetilde{Y} \to Y$ 
of a smooth scheme $Y$ along a smooth closed subscheme $V \subset Y$ is
\begin{equation}\label{K0blowup2}
[\widetilde{Y}] =[Y] + \sum_{k=1}^{\codim_Y(V)-1} [V] \,\,\L^k.
\end{equation}
\end{Lemma}

\proof We can write the class of the exceptional
divisor as $[E] = [V] ([\P^{\codim_Y(V)-1}] -1)$. Using
$\sum_{k=1}^{\codim_Y(V)-1} \L^k =  [\P^{\codim_Y(V)-1}] -1$ one obtains 
$[\widetilde{Y}] = [Y] +  [V] \sum_{k=1}^{\codim_Y(V)-1} \L^k$.
\endproof

In particular, through the motivic Euler characteristic, the image in $K_0(\cM)$
of the class in $K_0(\cV)$ is equal to $\chi([\tilde Y])=\chi_{\rm mot} (m(\tilde Y))$, 
so that the formula \eqref{K0blowup2} matches exactly the form of the
corresponding \eqref{Chowblowup}  and \eqref{Voevblowup}.

We then obtain the following explicit formula for the class in the Grothendieck
ring of the wonderful compactifications $\overline{Conf}_\Gamma(X)$.

\begin{Proposition}\label{GrclassConf}
Let $X$ be a quasi-projective variety and let $[X]$ denote its class in
the Grothendieck ring of varieties $K_0(\cV)$. Then,
for a given graph $\Gamma$, the class $[\overline{Conf}_\Gamma(X)]$ in 
$K_0(\cV)$ is given by 
\begin{equation}\label{classCGX}
[\konj{Conf}_\GG(X)] =[X]^{|\BV_\Gamma|} +
\sum_{\cN\in \cG_\Gamma\text{-nests}} \,\, 
[X]^{|\BV_{\Gamma/\delta_{\cN}(\Gamma)}|} \sum_{\mu \in M_\cN} \L^{\|\mu\|}.
\end{equation}
\end{Proposition}

\proof One can once again argue as in Lemma 3.1 of \cite{li}, using \eqref{K0blowup2}
instead of \eqref{Chowblowup} of \eqref{Voevblowup}, and obtain the analog of 
\eqref{recXNk}, with the same notation as in Proposition \ref{thmMTvoe}, namely
\begin{equation}\label{recXNkGr}
[X_\cN^{(k+1)}] = [X_\cN^{(k)}] + \sum_{\ell =1}^{r_{k,\cN}-1} [X_{\cN'}^{(k)}] \,\, \L^\ell =
[X_\cN^{(k)}] +  [X_{\cN'}^{(k)}]  ([\P^{r_{\gamma,\cN}-1}]-1).
\end{equation}
One then uses the same downward induction argument of Theorem 3.1 of \cite{li},
applying \eqref{recXNkGr} at each step and one obtains \eqref{classCGX}.
\endproof

Thus, for example, in the case of $X=\P^d$ we have the following formula for the
class in the Grothendieck ring.

\begin{Corollary}\label{GrPdConf}
For $X=\P^d$, the class $[\konj{Conf}_\GG(\P^d)]$ in  $K_0(\cV)$ is
\begin{equation}\label{classK0Pd}
[\konj{Conf}_\GG(\P^d)]=(\sum_{\ell=0}^d \L^\ell)^{|\BV_\Gamma|} +
\sum_{\cN\in \cG_\Gamma\text{-nests}} \,\,  (\sum_{\ell=0}^d \L^\ell)^{|\BV_{\Gamma/\delta_{\cN}(\Gamma)}|} \sum_{\mu \in M_\cN} \L^{\|\mu\|}.
\end{equation}
\end{Corollary}

The class in the Grothendieck ring can be written also in terms of the stratification.
This leads to interesting identities for the spaces $\overline{Conf}_\Gamma(X)$,
similar to the combinatorial identities proved in \cite{Cheah} in the 
Fulton--MacPherson case.

\begin{Lemma}\label{classK0strata}
The expression \eqref{classCGX} for the class $[\konj{Conf}_\GG(X)] $ in $K_0(\cV)$
can be equivalently written as
\begin{equation}\label{K0strataConfG}
 [\overline{Conf}_\Gamma(X)] = [Conf_\Gamma(X)] + \sum_{\cN \,\, \cG-nests} [X_{\cN}^\circ],
\end{equation}
where the $X_{\cN}^\circ$ are the open strata of \eqref{opstr}. This can then be written
equivalently as
\begin{equation}\label{K0strataConfG2}
 [\overline{Conf}_\Gamma(X)] = [Conf_\Gamma(X)] + \sum_{\cN \,\, \cG-nests} [Conf_{\Gamma/\delta_{\cN}(\Gamma)}(X)] \prod_{\gamma \in \BV_{\cT(\cN)}} 
 [ C_{\gamma/\delta_{\cN}(\gamma)} (\A^d) ],
\end{equation}
\end{Lemma}

\proof 
The stratification \eqref{opstr} of $\konj{Conf}_\GG(X)$ described in \S \ref{strataSec} also
gives us a way to compute the class in the Grothedieck ring. In fact, by the inclusion-exclusion
relation in the Grothedieck ring, the disjoint union 
$$ 
\overline{Conf}_\Gamma(X) = Conf_\Gamma(X) \cup \bigcup_{\cN \,\, \cG-nests} X_{\cN}^\circ  
$$
of the open strata corresponds to a sum of classes \eqref{K0strataConfG}.
We check the compatibility of \eqref{K0strataConfG2} with the formula \eqref{K0strataConfG}.

Recall that the open stratum $X_{\cN}^\circ$ is a fiber bundle over a base given by 
the configuration
space $Conf_{\Gamma/\delta_\cN(\Gamma)}(X)$, with fiber $\cF_\cN$ that is obtained as an
iteration of bundles, each with fiber $\cF_\gamma$ the space of translations and homothety
classes of distinct labeled  $|\BV_{\gamma/\delta_{\cN}(\gamma)}|$-tuples of points in $\bT_x$.
Thus, each $\cF_\gamma$ is identified with an open subvariety of the space $\P(\bT_x^{\BV_{\gamma/\delta_{\cN}(\gamma)}}/\bT_x)$ of screen configurations, isomorphic to the
configuration space $C_{\gamma/\delta_{\cN}(\gamma)} (\A^d)$.
Thus, by \eqref{trivFibr}, we can write each class in \eqref{K0strataConfG} as
$[X_{\cN}^\circ]=[Conf_{\Gamma/\delta_\cN(\Gamma)}(X)][\cF_\cN]$, where the class of the fiber
$[\cF_\cN]$ in turn can be written as a product $\prod_\gamma [\cF_\gamma]$ over
the graphs $\gamma$ in the forest $\cT$ of nested subgraphs with $\cN=\cN(\cT)$,
as in Theorem \ref{bundlestrata}, with each $[\cF_\gamma]=[C_{\gamma/\delta_{\cN}(\gamma)}(\A^d)]$.
\endproof

By comparing the two formulae \eqref{K0strataConfG} and  \eqref{classCGX}, we obtain
some explicit combinatorial identities involving the classes of the configuration
spaces $C_{\gamma/\delta_{\cN}(\gamma)}(\A^d)$, with $\gamma$ ranging over the 
graphs decorating the vertices of the 
 the forest $\cT$ of nested subgraphs for a given $\cG_\Gamma$-nest $\cN=\cN(\cT)$, and
 the classes of the projective spaces $\P^{r_{k,\cN}-1}$, with $r_{\gamma,\cN}$ as in
 \eqref{rgamma}.
 
 \begin{Lemma}\label{idLmuK0}
 For a given graph $\Gamma$, consider a $\cG_\Gamma$-nest $\cN$. For $\gamma$ in $\cG_\Gamma$ let $r_{\gamma,\cN}$ and $\mu=(\mu_\gamma)_{\gamma\in \cG_\Gamma}\in M_{\cN}$ be as in \eqref{rgamma} and \eqref{MN}. The we have in $K_0(\cV)$ the identity
 \begin{equation}\label{idLmu}
\sum_{\mu\in M_\cN} \L^{\|\mu \|} = \prod_{\gamma \in \cG_\Gamma} ([\P^{r_{\gamma,\cN}-1}]-1)
= \prod_{\gamma\in \cG_\Gamma} \frac{\L^{r_{\gamma,\cN}}-1}{\L-1}. 
 \end{equation}
 \end{Lemma}
 
 \proof Each class $[\P^{r_{\gamma,\cN}-1}]-1=\sum_{\ell=1}^{r_{\gamma,\cN}-1} \L^\ell=(\L^{r_{\gamma,\cN}}-1)(\L-1)^{-1}$. Thus, their product is simply 
$$ \sum_{\mu=(\mu_\gamma): 1\leq \mu_\gamma \leq r_{\gamma,\cN}-1} \L^{\sum_\gamma \mu_\gamma} =\sum_{\mu\in M_\cN} \L^{\|\mu\|}, $$
where $\sum_\gamma \mu_\gamma=\|\mu\|$ as in \eqref{normmu}.
\endproof

Moreover, we can express the classes of $X^{\BV_\Gamma}$ and $X^{\BV_{\Gamma/\delta_{\cN}(\Gamma)}}$ in the formula \eqref{classCGX} in terms of the classes of $Conf_\Gamma(X)$ and
$Conf_{\Gamma/\delta_{\cN}(\Gamma)}(X)$  of \eqref{K0strataConfG} in the following way.

\begin{Lemma}\label{XVandConfXK0}
For a given graph $\Gamma$ and a given $\cG_\Gamma$-nest $\cN$, we have the following
identities in $K_0(\cV)$:
\begin{equation}\label{XVConfGamma}
[X^{\BV_\Gamma}] =[Conf_\Gamma(X)]+ \sum_{\cN \in \cG_\Gamma\text{-nests}} [Conf_{\Gamma/\delta_{\cN}(\Gamma)}(X)],
\end{equation}
\begin{equation}
[X^{\BV_{\Gamma/\delta_{\cN}(\Gamma)}}] = [Conf_{\Gamma/\delta_{\cN}(\Gamma)}(X)]+
\sum_{\cN'\in \cG_\Gamma\text{-nests}\,:\,\, \cN \subset \cN'}
 [Conf_{\Gamma/\delta_{\cN'}(\Gamma)}(X)] .
\end{equation}
\end{Lemma}

\proof The first identity is an immediate consequence of the stratification of
$\overline{Conf}_\Gamma(X)$ by open strata $X_\cN^\circ$ and the fact that,
under the projection $\pi: \overline{Conf}_\Gamma(X) \to X^{\BV_\Gamma}$
the $X_\cN^\circ$ map to the $Conf_{\Gamma/\delta_{\cN}(\Gamma)}(X)$,
together with the additivity of Grothendieck classes over disjoint unions.
The second statement follows in the same way, with $\Gamma$ replaced by
its quotient $\Gamma/\delta_{\cN}(\Gamma)$, with the observation that the
$\cG_\Gamma$-nests $\cN'$ for $\Gamma$ that contain the nest $\cN$
can be identified with the $\cG_{\Gamma/\delta_{\cN}(\Gamma)}$-nests,
after identifying $$ \Gamma/\delta_{\cN'}(\Gamma) =
(\Gamma/\delta_{\cN}(\Gamma))/\delta_{\cN'}(\Gamma/\delta_{\cN}(\Gamma)). $$
\endproof

We then obtain the following identity.

\begin{Proposition}\label{idK0Conf}
The following identity holds between the classes of the configuration spaces
$C_{\gamma/\delta_{\cN}(\gamma)}(\A^d)$:
\begin{equation}\label{sumNclasses}
\begin{array}{ll}
\displaystyle{\sum_{\cN \in \cG_\Gamma\text{-nests}}   [Conf_{\Gamma/\delta_{\cN}(\Gamma)}(X)]} 
& \displaystyle{\prod_{\gamma \in \BV_{\cT(\cN)}} [C_{\gamma/\delta_{\cN}(\gamma)}(\A^d)]}  = \\[3mm] 
\displaystyle{\sum_{\cN \in \cG_\Gamma\text{-nests}} [Conf_{\Gamma/\delta_{\cN}(\Gamma)}(X)]} & \displaystyle{\left( 1+
\sum_{\cN' \in \cG_\Gamma\text{-nests}\,:\, \cN'\subset \cN} \prod_{\gamma \in \cG_\Gamma}
([\P^{r_{\gamma,\cN'}-1}]-1) \right).}
\end{array}
\end{equation}
\end{Proposition}

\proof Using Lemmata \ref{classK0strata}, \ref{idLmuK0}, and \ref{XVandConfXK0}, we 
obtain an identity
$$ [Conf_\Gamma(X)] + \sum_{\cN \in \cG_\Gamma\text{-nests}}   [Conf_{\Gamma/\delta_{\cN}(\Gamma)}(X)] \prod_{\gamma \in \BV_{\cT(\cN)}} [C_{\gamma/\delta_{\cN}(\gamma)}(\A^d)]  $$
$$ = [Conf_\Gamma(X)]+ \sum_{\cN \in \cG_\Gamma\text{-nests}} [Conf_{\Gamma/\delta_{\cN}(\Gamma)}(X)] + $$
$$ \sum_{\cN \in \cG_\Gamma\text{-nests}}  \left( [Conf_{\Gamma/\delta_{\cN}(\Gamma)}(X)]+
\sum_{\cN'\in \cG_\Gamma\text{-nests}\,:\,\, \cN \subset \cN'}
 [Conf_{\Gamma/\delta_{\cN'}(\Gamma)}(X)] \right) 
\prod_{\gamma \in \cG_\Gamma} ([\P^{r_{\gamma,\cN}-1}]-1) . $$
We subtract the $[Conf_\Gamma(X)]$ on both sides and rearrange and reindex
the terms in the second summation on the right hand side in such a way that each 
$\cG_\Gamma$-nest $\cN$ appears once in the
summation, with the corresponding class $[Conf_{\Gamma/\delta_{\cN}(\Gamma)}(X)]$
multiplied by the sum of the classes $\prod_{\gamma \in \cG_\Gamma} ([\P^{r_{\gamma,\cN'}-1}]-1)$, one for each $\cN' \subset \cN$. There is then an additional $+1$ term coming
from the single contribution of a class $[Conf_{\Gamma/\delta_{\cN}(\Gamma)}(X)]$
from the first summation on the right hand side of the formula above. This gives the
formula on the right hand side of \eqref{sumNclasses}.
\endproof

\subsection{Mixed Hodge structures and virtual Hodge polynomials}\label{HodgeSec}

The discussion of the motives in the previous section can also be adapted to
working with Hodge polynomials and mixed Hodge structures, instead of
classes in the Grothendieck ring.

In the case of the Fulton--MacPherson compactification, the mixed Hodge structures and
Hodge polynomials were computed explicitly in \cite{Cheah}, \cite{Getzler}, \cite{Manin}.
In particular, in that case, one knows that there is a nice way to write a generating 
function for the Hodge polynomials. In our case we do not get as explicit an answer,
but we can see that the relation of Proposition \ref{idK0Conf} provides a partial
analog in our setting.   

We recall that the virtual Hodge polynomial of an algebraic variety is defined as 
\begin{equation}\label{HodgePoly}
e(X) (x,y)= \sum_{p,q=0}^d e^{p,q}(X) x^p y^q, \ \ \ \text{ with } \ \ \ 
e^{p,q}(X) =\sum_{k=0}^{2d} (-1)^k h^{p,q}(H^k_c(X)),    
\end{equation}
where for each pair of integers $(p,q)$ the $h^{p,q}(H^k_c(X))$ are the Hodge numbers
of the mixed Hodge structure on the cohomology with compact support of $X$. If $X$ is
smooth projective, then the virtual Hodge polynomial reduces to the Poincar\'e polynomial, 
with $e^{p,q}(X)=(-1)^{p+q}h^{p,q}(X)$ the classical pure Hodge structure. 
It is well known that the virtual Hodge polynomial is, like the Euler characteristic, a
motivic invariant in the sense recalled at the beginning of \S \ref{GrclassSec} above,
namely it factors through the Grothendieck ring $K_0(\cV)$.

This means that, having an explicit formula for the class of a variety in the
Grothendieck ring, one can use it to compute the virtual Hodge polynomial.
The computation of the classes in the Grothendieck ring of varieties we obtained
in the previous section then gives us a formula for the Hodge polynomials 
of the graph configuration spaces we are considering here.

\begin{Proposition}\label{Hodgepoly}
The virtual Hodge polynomial $e(\overline{Conf}_\Gamma(X))(x,y)$ is given, as a
function of $e(X)(x,y)$, by the formula
\begin{equation}\label{eConfX}
e(\overline{Conf}_\Gamma(X)) = e(X)^{|\BV_\Gamma|} +\sum_{\cN} e(X)^{|\BV_{\Gamma/\delta_\cN(\Gamma)}|} \prod_{\gamma \in \cG_\Gamma} (e(\P^{r_{\gamma,\cN}-1}) -1).
\end{equation}
Moreover, the Hodge polynomials $e(\overline{Conf}_{\Gamma/\delta_{\cN}(\Gamma)}(X))(x,y)$
and $e(C_{\gamma/\delta_{\cN}(\gamma)}(\A^d))(x,y)$ satisfy the relation
\begin{equation}\label{eConfXAd}
\begin{array}{ll}
\displaystyle{\sum_{\cN \in \cG_\Gamma\text{-nests}}   e(Conf_{\Gamma/\delta_{\cN}(\Gamma)}(X))} 
& \displaystyle{\prod_{\gamma \in \BV_{\cT(\cN)}} e(C_{\gamma/\delta_{\cN}(\gamma)}(\A^d))}  = \\[3mm] 
\displaystyle{\sum_{\cN \in \cG_\Gamma\text{-nests}} e(Conf_{\Gamma/\delta_{\cN}(\Gamma)}(X))} & \displaystyle{\left( 1+
\sum_{\cN' \in \cG_\Gamma\text{-nests}\,:\, \cN'\subset \cN} \prod_{\gamma \in \cG_\Gamma}
(e(\P^{r_{\gamma,\cN'}-1})-1) \right).}
\end{array}
\end{equation}
\end{Proposition}

\proof The result follows directly from the Grothendieck ring calculations in \eqref{classCGX}
and \eqref{sumNclasses}, using the fact that the virtual Hodge polynomial defines a ring homomorphism $e: K_0(\cV) \to \Z[x,y]$.
\endproof

\section{Residues of Feynman integrals}

We now consider the Feynman integrals in configuration space and discuss the
relevance of the motivic point of view discussed in the previous sections.
The regularization and renormalization of Feynman amplitudes in configuration
space, using the wornderful compactifications of \cite{DecoPro}, \cite{li2}, was
recently analyzed in the paper of Bergbauer--Brunetti--Kreimer \cite{BerBruKr}.
Here we concentrate on the {\em residues} of the Feynman amplitudes, rather
than on their renormalized values as in \cite{BerBruKr}. We then find a setting that
parallels to some extent the analysis in terms of Hodge structures for the
Feynman amplitudes in momentum space given in \cite{BK}.

\begin{Definition}\label{logdivprimdef}
Suppose given an underlying variety $X$ of dimension $d=\dim X$.
A connected graph $\Gamma$ is logarithmically divergent (log divergent) if
it satisfies the condition
\begin{equation}\label{logdiv}
d \, b_1(\Gamma) = 2 \, |\BE_\Gamma|,
\end{equation}
or equivalently (for connected graphs)
\begin{equation}\label{logdiv2}
 (d-2)\, | \BE_\Gamma| = d\, (|\BV_\Gamma| -1),
\end{equation} 
and all subgraphs $\gamma \subseteq \Gamma$ satisfy
\begin{equation}\label{subgrlogdiv}
d \, b_1(\gamma) \leq 2 \, |\BE_\gamma|,
\end{equation}
which for a connected subgraph means $(d-2) |\BE_\gamma| \leq  d |\BV_\gamma| -d$.
A subgraph $\gamma \subseteq \Gamma$ is divergent if it satisfies
$d \, b_1(\gamma) = 2 \, |\BE_\gamma|$.
A primitive graph is a log divergent graph that contains no divergent subgraphs.
A graph with $d \, b_1(\Gamma) > 2 \, |\BE_\Gamma|$ is said to have worse 
than logarithmic divergences. For connected graphs this corresponds to
$(d-2) |\BE_\Gamma| > d \, |\BV_\Gamma| -d$.
\end{Definition}

In the four-dimensional case $d=4$ the log divergent condition recovers
the usual condition that the graph has $n$ loops and $2n$ edges. Renormalization
in momentum space for graphs with worse than logarithmic divergences was
considered from a Hodge theoretic point of view in \cite{BK}.

In \cite{BerBruKr}, the regularization and renormalization of Feynman integrals
in configuration spaces is obtained in the primitive case by a simple subtraction,
whereby the Feynman density is pulled back to the wonderful compactification
$\overline{Conf}_\Gamma(X)$ and regularized there to a meromorphic function
of a complex parameter $s$ with a pole at $s=1$, whose residue is supported on
the exceptional divisor of the blowup. This is then subtracted (local minimal subtraction)
and the resulting density is pushed forward to a regular density on $X^{\BV_\Gamma}$ 
whose value at $s=1$ is the renormalized density (see Theorem 3.1 of \cite{BerBruKr}). 
The case of log-divergent, non-primitive graphs is more complicated because 
the stratification of the exceptional divisor of the blowup plays an important 
role and the regularization and renormalization procedure is given by a local
minimal subtraction in every factor of a product indexed over $\cG$-nests, 
see Theorem 5.3 of \cite{BerBruKr}.

\subsection{Weights of Feynman graphs}\label{weightsSec}

We consider as above a (quasi)projective variety $X$ of dimension $d=\dim X$.
We write $X(\C)$ for its complex points and $M=X(\R)$ for its real part. In particular,
one can consider the case where  $X = \P^d(\C)$ and $M=\P^d(\R)$, as a compactification 
of the $d$-dimensional spacetime $\R^d$. 

We also consider a scalar quantum field theory where the Lagrangian has a 
potential $U$ given by a polynomial in field $\phi$,
\begin{eqnarray*}
U = \sum_{k=1}^s U_k \phi^k.
\end{eqnarray*}

Let then $\GG$ be a connected Feynman diagram of the quantum field 
theory having no multiple edges or tadpoles (looping edges). 
Let $Conf_\Gamma(X)$ and $\overline{Conf}_\Gamma(X)$ be the
configuration space and its wonderful compactification, as in the
previous sections. We also consider $Conf_\Gamma(M)$ and 
$\overline{Conf}_\Gamma(M)$, which are the real
loci of $Conf_\Gamma(X)$ and $\overline{Conf}_\Gamma(X)$, 
respectively.

\begin{rem}\label{reallocus}
Notice that the real locus we consider here is not the ``real blowup" of
$M^{\BV_\Gamma}$ in the sense of \cite{AxSing2} and \cite{BerBruKr},
but the real locus of the complex blowup $\overline{Conf}_\Gamma(X)$
of the complex manifold $X^{\BV_\Gamma}$. The real blowup, as 
shown in \cite{AxSing2}, is a real manifold with corners, hence it defines a
chain with boundary. The real locus of the complex blowup is a real 
algebraic variety. Thus, it defines a
middle dimensional {\em cycle} in the complex variety. However,
the real variety $\overline{Conf}_\Gamma(M)$ may be non-orientable.
\end{rem}

Feynman rules assign a weight to a graph $\GG$ as follows: 

\begin{itemize}
\item The vertices are labelled by the coordinates $x_1,\dots,x_n$ of $Conf_\Gamma(M)$.
\item To each edge with $\dd_\GG (e) = \{x,y\}$, one assigns a massless Euclidean
propagator 
\begin{equation}\label{propagator}
G(x-y) = i  \left( \frac{1}{(x-y)^2} \right)^{\frac{d-2}{2}}
\end{equation}
These are rational functions when the dimension $d$ is even. 
\item The {\it (unregularized)  weight} of the graph $\Gamma$ is defined as
\begin{eqnarray}
^0W_\GG := \int_{M^{\BV_\GG}}  \Go_\GG
 \label{eqn_weights}
\end{eqnarray}
where $M^{\BV_\Gamma}=X^{\BV_\Gamma}(\R)$ is the real locus, and
\begin{equation}\label{omega0dens}
 \Go_\GG := \prod_{v \in \BV_\GG} U_{|v|} \times 
\prod_{{ \dd_\GG (e) = \{v_e.v^e\}} \atop {e \in \BE_\GG} } G(x_{v_e} - x_{v^e}) 
\prod_{v \in\BV_\GG} dx_v,
\end{equation}
where $U_{|v|}$ is the coefficient of the monomial $\phi^k$ in the potential $U$ with
$k=|v|$ the valence of the vertex $v$.
\end{itemize}

\subsection{Graph hypersurfaces and divergences}
Let $\pi_\Gamma:  \overline{Conf}_\Gamma(X) \to X^{\BV_\GG}$ be the rational map 
inductively constructed in \S \ref{blowupSec} as iterated blowups, and let $\pi_\Gamma^\R:  \overline{Conf}_\Gamma(M) \to M^{\BV_\GG}$ be its restriction to the real locus. 

\begin{Lemma}\label{divMGamma}
The divergent locus of the density $\Go_\GG$ of \eqref{omega0dens} in $M^{\BV_\Gamma}$ is given by the union of diagonals $\bigcup_{e \in \BE_\GG} \GD_e$.
\end{Lemma}

\proof For massless Euclidean field theories, 
the {\it graph hypersurface}  of $\GG$ (that is, the pole locus $\{\Go_\GG = \infty\}$ in 
$X^{\BV_\GG}$) is simply the union of quadrics
\begin{eqnarray}
Z_\GG := \left\{\prod_{{ \dd_\GG (e) = \{v_e.v^e\}} \atop {e \in \BE_\GG} } (x_{v_e} - x_{v^e})^2
= 0 \right\}.
\label{eqn_g_hyp}
\end{eqnarray}
The defining equation (\ref{eqn_g_hyp}) of $Z_\GG$ is a real polynomial with
non-negative values on real points, hence the intersection $Z_\GG(\C) \bigcap M[\GG]$
is given by $x_{v_e} = x_{v^e}$ i.e., it is the union of diagonals 
$\bigcup_{e \in \BE_\GG} \GD_e \subset  M^{|\BV_\GG|}$.
\endproof

\subsection{Order of poles in the blowups}

In the following, assuming $d$ even, we use the notation
\begin{equation}\label{fe}
f_e (x) = (x_{v_1} - x_{v_2}) \ \ \ \text{ for } \ \  \{ v_1, v_2 \} = \dd(e),
\end{equation}
so that the propagator $G$ of
\eqref{propagator} is given by $G(x_{v_1} - x_{v_2})= f_e^{2-d}(x)$, 
as in \eqref{omega0dens}. The function $f_e$ is also the defining
function of the diagonal $\Delta_e=\{ f_e =0 \}$, which is a codimension 
$d$ subvariety in $X^{\BV_\Gamma}$.

\begin{Proposition}
\label{prop_form}
Let $\Gamma$ be a primitive, biconnected, log divergent graph. Then
the proper transform $\widetilde{\Go}_\GG=\pi_\Gamma^*(\omega_\Gamma)$ 
of the form $\Go_\GG$ of \eqref{omega0dens} to the blowup of $X^{\BV_\Gamma}$
along the deepest diagonal $\Delta_\Gamma$ has a pole of order one along
the exceptional divisor, while the pullback to the blowups along (the dominant 
transforms of) the (poly)diagonals $\Delta_\gamma$, with $\gamma \subset \Gamma$
have no other poles along the exceptional divisors $E_\gamma$.  
\end{Proposition}

\begin{proof} In the model case of a 
coordinate linear space $L$ defined by equations 
$\{z_1= \cdots = z_p =0\} \subset \C^{d \, |\BV_\GG|}$, one can choose coordinate
charts in the blowup with coordinates $w_i$, so that $w_i = z_i$ for $i=p,\ldots,  d |\BV_\GG|$
and $w_i w_p= z_i$ for $i<p$, so that, in these coordinates, the exceptional divisor is defined 
by $w_p=0$. Thus, the orientation form satisfies
\begin{eqnarray} 
\pi^*(dz_1 \wedge \cdots \wedge d z_{d |\BV_\GG|}) &=& d(w_p w_1) \wedge \cdots \wedge d(w_p w_{p-1})
 						\wedge dw_p\wedge \cdots \wedge d w_{d |\BV_\GG|}  \nonumber \\
						&=& w^{p-1}_p dw_1 \wedge \cdots \wedge d w_{p-1} 
						\wedge dw_p\wedge \cdots \wedge d w_{d |\BV_\GG|}.  \nonumber
\end{eqnarray}
This has a zero of order ${\rm codim}(L) -1$ along the exceptional divisor of the blowup.

Let $\Delta_\gamma$ be the diagonal associated to a connected subgraph 
$\gamma \subset \Gamma$. One obtains a minimal set of equations defining
$\Delta_\gamma$ by choosing a spanning tree $\tau$ for $\gamma$. Then
\begin{equation}\label{spantreeDeltag}
\Delta_\gamma =\{ f_e =0 \, |\, e\in \BE_\tau \},
\end{equation}
with $f_e$ as in \eqref{fe}. This gives ${\rm codim}(\Delta_\gamma) = d \, |\BE_\tau|$.
For a spanning tree we have $|\BE_\tau|=|\BV_\gamma| -1$, since $\gamma$ is
connected, so this gives $\dim(\Delta_\gamma)= d\, (|\BV_\Gamma|-|\BV_\gamma|+1)$,
as we saw in the previous sections.

The form $\omega_\Gamma$ of \eqref{omega0dens} has order of pole along $\Delta_\gamma$
given by
\begin{equation}\label{poleomegaDg}
{\rm ord}_\infty(\omega_\Gamma,\Delta_\gamma) = (d-2) |\BE_\gamma|,
\end{equation}
coming from the factors $f_e^{2-d}$ with $e\in \BE_\gamma$.

When we consider the blowup $\pi_\gamma: Bl_{\Delta_\gamma}(X^{\BV_\Gamma}) \to X^{\BV_\Gamma}$ and we pull back the singular differential form $\omega_\Gamma$ along $\pi_\gamma$, we obtain a form $\pi_\gamma^* (\omega_\Gamma)$ that has order of pole along the exceptional
divisor $E_\gamma$ of the blowup given by
\begin{equation}\label{poleomegaDg}
{\rm ord}_\infty(\pi^*_\gamma(\omega_\Gamma),E_\gamma) = (d-2) |\BE_\gamma| - d |\BE_\tau| +1
= (d-2) |\BE_\gamma| - d (|\BV_\gamma| -1) +1.
\end{equation}

If the graph $\Gamma$ is a primitive, biconnected, log divergent graph, then $(d-2) |\BE_\Gamma| = d |\BV_\Gamma| -d$ and $\Gamma$ contains no divergent subgraphs, so that
$(d-2) |\BE_\gamma| < d |\BV_\gamma| -d$,  for all subgraphs $\gamma \subset \Gamma$.
Thus, in this case, the pullback $\pi_\Gamma^*(\omega_\Gamma)$ along the map that corresponds to the blowup along the deepest diagonal $\Delta_\Gamma$ 
has a pole of order one along the exceptional divisor, while all the further blowups along
the dominant transforms of the $\Delta_\gamma$ do not contribute any poles.
\end{proof}

This corresponds to the case analyzed in Theorem 3.1 of \cite{BerBruKr}, where one needs 
just one pole subtraction in order to renormalize the Feynman amplitude. Here it comes from
the subtraction of the simple pole along the exceptional divisor $E_\Gamma$ of the
blowup along the deepest diagonal $\Delta_\Gamma$. 

In the case where $\Gamma$ is log divergent but no longer primitive, the pullback of $\omega_\Gamma$
to the blowups along (the dominant transforms of) the $\Delta_\gamma$ with
$(d-2) |\BE_\gamma| = d (|\BV_\gamma| -1)$ has a pole of order one along the exceptional
divisor $E_\gamma$. 
This is the more general situation analyzed in \S 5 of \cite{BerBruKr}.
 
In the even more general case of graphs $\Gamma$ that have worse than logarithmic
singularities, one finds that the order of pole along the exceptional divisors of
the iterated chain of blowups that define the wonderful model $\overline{Conf}_\Gamma(X)$
is given by the following.

\begin{Corollary}\label{worselog}
Let $\Gamma$ be a connected graph which has worse than logarithmic divergences. Then
for every connected induced subgraph $\gamma \subset \Gamma$ that has 
$(d-2) |\BE_\Gamma| > d \, |\BV_\Gamma| -d$, the pullback
$\pi_\gamma^*(\omega_\Gamma)$ of the form $\omega_\Gamma$ of \eqref{omega0dens}
to the blowup along the (dominant transform of) $\Delta_\gamma$ has poles of higher order
\begin{equation}\label{ordmore1}
{\rm ord}_\infty (\pi_\gamma^*(\omega_\Gamma), E_\gamma) = 
(d-2) |\BE_\gamma| - d (|\BV_\gamma| -1) +1 > 1
\end{equation}
along the exceptional divisors $E_\gamma$ in the blowup.
\end{Corollary}

\subsection{The Poincar\'e residue}

We discuss briefly the residues of Feynman amplitudes, first in the primitive and the log divergent case and then
in the more general case of graphs with worse than logarithmic divergences. We want to
remain within the setting of algebraic varieties and periods, hence we describe the residues
of Feynman amplitudes in terms of Poincar\'e residues and Hodge structures. 

We recall the basic definition of the Poincar\'e residue of a differential
form with simple poles along a hypersurface (see \cite{GrHa}, p.147).
Given a hypersurface $Y$ in a smooth $n$-dimensional 
projective variety $X$, locally defined by an equation $\{ f (z) =0 \}$,  
an $n$-form 
\begin{equation}\label{omegagf}
 \omega = \frac{g(z) \, dz_1 \wedge \cdots \wedge dz_n}{f(z)}  \in \Omega^n(X) 
\end{equation} 
can always be written as 
\begin{equation}\label{omegafdf}
 \omega = \frac{df}{f} \wedge \omega', 
\end{equation}
where $\omega'$ can be taken of the form
$$ \omega' =(-1)^{i-1} \frac{g(z)\, dz_{[i]}}{\frac{\partial f}{\partial z_i}}, $$
for any $i$ such that $\frac{\partial f}{\partial z_i} \neq 0$, with
$$ d z_{[i]} = dz_1 \wedge \cdots \wedge dz_{i-1} \wedge dz_{i+1} \wedge\cdots
\wedge dz_n. $$
The Poincar\'e residue of $\omega$ is then the $(n-1)$-form on $V$ defined by
\begin{equation}\label{PoiRes}
{\rm Res}[\omega] =\left. (-1)^{i-1} \frac{g(z)\, dz_{[i]}}{\frac{\partial f}{\partial z_i}} \right|_{f=0} \in \Omega^{n-1}(Y).
\end{equation}

\begin{Proposition}\label{uniqueRes}
Let $\Gamma$ be a biconnected, primitive, log divergent graph. Then the pullback of
the differential form $\omega_\Gamma$ of \eqref{omega0dens} to the wonderful model 
$\overline{Conf}_\Gamma(X)$ has a unique residue, 
which is a $(d|\BV_\Gamma|-1)$-form on the exceptional divisor $E_\Gamma$ of the
blowup of the deepest diagonal, 
\begin{equation}\label{Restildeomega}
{\rm Res}[\pi^*(\omega_\Gamma)]  \in \Omega^{d \, |\BV_\Gamma| -1}(E_\Gamma).
\end{equation}
\end{Proposition}

\proof
We have seen in Proposition \ref{prop_form} that, in the case of a primitive graph $\Gamma$,
the pullback $\tilde \omega_\Gamma =\pi^*_\Gamma(\omega_\Gamma)$ to the blowup of
$X^{\BV_\Gamma}$ along the deepest diagonal $\Delta_\Gamma \simeq X$, is a differential
form as in \eqref{omegagf}, with a simple pole along the exceptional divisor $E_\Gamma$, with\,\, $f=0$ the defining equation of $E_\Gamma$. Therefore $\tilde \omega_\Gamma$ can be rewritten in the form \eqref{omegafdf} and it has a well defined Poincar\'e residue ${\rm Res}[\tilde\omega_\Gamma]$, which is a $(d|\BV_\Gamma|-1)$-form on the exceptional divisor $E_\Gamma$.
The successive blowups along the dominant transforms of
the $\Delta_\gamma$, for $\gamma$ ranging over $\cG_\Gamma$,
do not contribute any further poles, since the graph has no
subdivergences. Moreover, because the order of the sequence of blowup is
determined by ordering $\cG_\Gamma$ in such a way that $i \leq j$ if $\gamma_i \supseteq \gamma_j$, so that $\Delta_{\gamma_i}\subseteq \Delta_{\gamma_j}$, any two diagonals
$\Delta_i$ and $\Delta_j$ that intersect along $\Delta_{\gamma_i\cup \gamma_j}$ have
dominant transforms that no longer intersect, once the blowup along $ \Delta_{\gamma_i\cup \gamma_j}$ has been performed already, and intersect transversely the exceptional
divisor of this blowup. Thus, after the first blowup along the deepest diagonal $\Delta_\Gamma$,
one obtains a residue 
${\rm Res}[\pi^*_\Gamma(\omega_\Gamma)] \in \Omega^{d |\BV_\Gamma| -1} (E_\Gamma)$. 
The pullback of this form along the successive blowups gives a $(d |\BV_\Gamma| -1)$-form supported on the dominant transform $E_\Gamma$ in $\overline{Conf}_\Gamma(X)$, which has zeros at the intersections of $E_\Gamma$ with the other exceptional
divisors $E_\gamma$.
\endproof

Consider next the case where the graph $\Gamma$ is logarithmically divergent, but not
primitive. Let $\cG_\Gamma$ be
ordered in such a way that  $i \leq j$ if $\gamma_i \supseteq \gamma_j$, as before, with
$\pi: \overline{Conf}_\Gamma(X)\to X^{\BV_\Gamma}$ the iterated blowups along the $\Delta_\gamma$, with $\gamma \in \cG_\Gamma$ in the assigned ordering.

This means that there are connected induced
subgraphs $\gamma \subset \Gamma$ for which the pullback to the blowup along 
(the dominant transform of) $\Delta_\gamma$ of the form $\omega_\Gamma$
has poles of order one along the exceptional divisor $E_\gamma$.  They are
precisely those satisfying the divergence condition.

Let us denote by $\cG_\Gamma^{\log}$
the subset $\cG_\Gamma^{\log} \subset \cG_\Gamma$ of 
subgraphs $\gamma$, satisfying the logarithmic divergence condition
\begin{equation}\label{divcond} 
\cG_\Gamma^{\log} :=\{ \gamma \in \cG_\Gamma\,|\, (d-2) |\BE_\gamma| = d (|\BV_\gamma| -1) \}.
\end{equation}
We then have the following result on the residues of the Feynman amplitude.

\begin{Proposition}\label{nonprimLogDivRes}
Let $\Gamma$ be a logarithmically divergent, non-primitive graph.  Then the pullback $\pi^*(\omega_\Gamma)$ of the form \eqref{omega0dens} has Poincar\'e residues along each divisor $E_\gamma$ for $\gamma \in \cG_\Gamma^{\log}$. Then the residue is given by a form
\begin{equation}\label{ResLogDiv}
{\rm Res}^\ell [\pi^*(\omega_\Gamma)] \in \Omega^{n-\ell} (E_{\gamma_1}\cap \cdots \cap E_{\gamma_\ell}).
\end{equation}
This is trivial unless the set $\cG_\Gamma^{\log}$ is a $\cG_\Gamma$-nest.
\end{Proposition}

\proof First notice that the form $\pi^*(\omega_\Gamma)$ has poles of order one 
along $E_\gamma$, for each $\gamma \in \cG_\Gamma$ 
satisfying $(d-2) |\BE_\gamma| = d (|\BV_\gamma| -1)$.   Thus, $\pi^*(\omega_\Gamma)$ is defined on
$\overline{Conf}_\Gamma(M) \smallsetminus \cup_{\gamma \in \cG_\Gamma^{\log}} E_\gamma$.
By iterating the procedure used to rewrite a form \eqref{omegagf} as \eqref{omegafdf}, one can
define {\em iterated residues} (see for instance Theorem 1.1 of \cite{ATYu}). 
For an $n$-form $\omega$ with a pole of order one along each component $Y_i$ of a hypersurface $Y=Y_1 \cup \cdots \cup Y_\ell$, where the $Y_i$ intersect transversely, the iterated residue
gives an $(n-\ell)$-form 
\begin{equation}\label{iterRes}
{\rm Res}^\ell[\omega] \in \Omega^{n-\ell}(Y_1 \cap \cdots \cap Y_\ell).
\end{equation}
We know that the intersection $\cap_{\gamma \in \cG_\Gamma^{\log}} E_\gamma$ is
non-empty if and only if the set $\cG_\Gamma^{\log}$ is a $\cG_\Gamma$-nest. Thus, one
obtains the residue \eqref{ResLogDiv}.
\endproof

In the more general case, where the graph has more than logarithmically divergent subgraphs, one has to deal with a form $\pi^*(\omega_\Gamma)$ that has poles of higher order along some of the exceptional divisors $E_\gamma$.

In affine space $\A^N$ a differential form 
$$ \omega = \frac{P(z)\, dz_1\wedge \cdots \wedge dz_N}{Q_1^{r_1}(z) \cdots Q_m^{r_m}(z)} $$
with poles of higher order $r_k$ along the hypersurfaces $Y_i$ defined by $Q_i =0$ is
cohomologous to a form with only poles of order one,
$$ \omega' = \sum_J \frac{P_J(z) \,   dz_1\wedge \cdots \wedge dz_N}{Q_{j_1}(z) \cdots Q_{j_k}(z)}, $$
with $J=\{ j_1, \ldots, j_k\}$, $k\leq N$. (See for instance Theorem 1.8 of \cite{ATYu}.)

This is not true in general for the complement of a hypersurface in a smooth projective 
variety, by using rational forms.  However, in the case of a smooth hypersurface $Y$ in $\P^n$, 
it was shown by Griffiths in \cite{Griff}
that there are Poincar\'e residues for forms with higher order poles. The $n$-forms 
$$ \omega = \frac{P(z)\, dz_1\wedge \cdots \wedge dz_n}{Q^{r+1}(z)} $$
with poles of order $r+1$ along the smooth hypersurface $Y=\{ Q=0 \}$ generate a
subspace $F^{n-r}H^n(\P^n\smallsetminus Y)$ of the cohomology $H^n(\P^n\smallsetminus Y)$
whose image under the Poincar\'e residue gives the pieces of the Hodge filtration on the 
primitive cohomology of the hypersurface,
$$ {\rm Res}(F^{n-r}H^n(\P^n\smallsetminus Y))= F^{(n-1-r)} H^{n-1}_{prim}(Y). $$

This result relating the pole filtration to the Hodge filtration was further generalized to
the complement of normal crossings divisors in smooth projective varieties by Deligne 
in \cite{Deligne} II \S 3.13, and the comparison between pole and Hodge filtration
for singular hypersurfaces was further analyzed by Deligne and Dimca in \cite{DelDim}
and Dimca and Saito in \cite{DimSa}.

Thus, if we momentarily ignored the other divergences coming from
the rest of $Z_\Gamma$ in $\overline{Conf}_\Gamma(X)$, we would
conclude that for a subgraph $\gamma \subseteq \Gamma$ that has worse than
logarithmic divergences, the pullback $\pi^*_\gamma(\omega_\Gamma)$ 
of the Feynman density \eqref{omega0dens} determines an element in
the polar filtration of the complement of the exceptional divisor $E_\gamma$
in $\overline{Conf}_\Gamma(X)$.  
Through Poincar\'e residues, this would then determine
an element in the Hodge filtration of the primitive cohomology of $E_\gamma$. 
The situation is in fact made more complicated by the presence of the additional
singularities coming from the hypersurface $Z_\Gamma$ of \eqref{eqn_g_hyp}.

\subsection{Regularization of contours by Leray coboundaries}

We propose here a regularization procedure for the divergent Feynman
amplitudes \eqref{eqn_weights}, where instead of regularizing the
form as in \cite{BerBruKr} we regularize the domain of integration using
Leray coboundaries, see \cite{Mar}.

Let $E_\gamma$ be one of the exceptional divisors in $\overline{Conf}_\Gamma(X)$
along which the pullback $\pi^*(\omega_\Gamma)$ of the Feynman amplitude 
\eqref{omega0dens} has poles (possibly of higher order). 

The unregularized Feynman weight \eqref{eqn_weights} is given by the
integral over the middle dimensional
cycle in $X^{\BV_\Gamma}$ given by the real locus
$\sigma= X^{\BV_\GG}(\R)=M^{\BV_\Gamma}$, see Remark \ref{reallocus}.

\begin{rem}\label{orientreal}
In the case of even dimensional spacetime,  the real locus 
$\overline{Conf}_\Gamma(\R^d)$ of the configuration space 
$\overline{Conf}_\Gamma(\A^d)$ is non-orientable. Thus, the
configuration spaces $\overline{Conf}_\Gamma(M)$ that
contain $\overline{Conf}_\Gamma(\R^d)$ are non-orientable.
However, in such cases, one can define the regularized weights 
in the same way that is described here below, after passing to a 
double cover of $\overline{Conf}_\Gamma(X)$, branched along 
$\bigcup_{\Gg \in \cG_{\Gamma}} E_\Gg$. The real locus of
this branched cover is orientable. With a slight abuse of notation, 
in the following we do not distinguish explicitly between 
$\overline{Conf}_\Gamma(M)$ and its orientable double cover.
\end{rem}

In particular, as we have seen in
Lemma \ref{divMGamma}, the divergences along the domain of integration
come from the real locus of $\cup_e \Delta_e$,
and in particular, within this locus, from the intersection $\sigma \cap \Delta_\gamma 
=\Delta_\gamma(\R)$, for $\gamma \subseteq \Gamma$ a divergent subgraph. 

Let $\tilde \sigma_\gamma = \pi^{-1}(\sigma \cap \Delta_\gamma) \subset E_\gamma$.
This is a $d |\BV_\Gamma|-1$-cycle in $E_\gamma$. The Leray coboundary
$\cL_\epsilon(\tilde \sigma_\gamma)$ of $\tilde \sigma_\gamma$ is a $d |\BV_\Gamma|$-cycle 
in $\overline{Conf}_\Gamma(X)$ obtained as follows. Let $\partial D_\epsilon(E_\gamma)$
be the boundary of a tubular neighborhood of radius $\epsilon$ around $E_\gamma$. This
is a circle bundle $\pi_\epsilon: \partial D_\epsilon(E_\gamma) \to E_\gamma$ over $E_\gamma$
and one sets $\cL_\epsilon(\tilde \sigma_\gamma) = \pi_\epsilon^{-1}(\tilde \sigma_\gamma)$.
The preimage $\overline{Conf}_\Gamma(M)=\pi^{-1}(\sigma) \subset 
\overline{Conf}_\Gamma(X)$ of the real locus $\sigma = M^{\BV_\Gamma}$ intersects 
$\partial D_\epsilon(E_\gamma)$ in its real points. 

Let then $\Sigma_\epsilon \subset \cL_\epsilon(\tilde\sigma_\Gamma)$ be
a deformation to $\partial D_\epsilon(E_\gamma)$ of $\overline{Conf}_\Gamma(M) 
\cap D_\epsilon(E_\gamma)$, with fixed $\overline{Conf}_\Gamma(M) \cap \partial 
D_\epsilon(E_\gamma)$.
If $\Sigma_\epsilon$ does not intersect the locus $\tilde Z_\Gamma \cap 
\cL_\epsilon(\tilde\sigma_\Gamma)$, where $\tilde Z_\Gamma= \pi_\Gamma^{-1}(Z_\Gamma)$ 
is the preimage of the graph hypersurface of \eqref{eqn_g_hyp} along which the 
form $\omega_\Gamma$ is singular,  one can regularize the integral 
$$ \int_{\overline{Conf}_\Gamma(M)} \pi^*(\omega_\Gamma) $$
by replacing the part
$$ \int_{\overline{Conf}_\Gamma(M) \cap D_\epsilon(E_\gamma)} \pi^*(\omega_\Gamma) $$
of the integral with an integration along the Leray coboundary 
\begin{equation}\label{regLeray}
 \int_{\Sigma_\epsilon} \pi^*(\omega_\Gamma).
\end{equation}

There is an ambiguity involved in the choice of this regularization of the domain
of integration, as in the choice of contours that avoid poles in the one dimensional
setting, which is measured in terms of residues.

\begin{Proposition}\label{ambiguity}
Let $\Gamma$ be a logarithmically divergent graph with $\gamma \subseteq \Gamma$
a divergent subgraph. Then the regularization \eqref{regLeray} is defined up to an
ambiguity measured by the integral
\begin{equation}\label{resambi}
2\pi i \int_{\tilde\sigma_\gamma} {\rm Res}[\pi^*_\gamma(\omega_\Gamma)]
\end{equation}
of the Poincar\'e residue ${\rm Res}[\pi^*_\gamma(\omega_\Gamma)]\in \Omega^{d |\BV_\Gamma|-1}(E_\gamma)$ along the cycle $\tilde\sigma_\gamma = \pi^{-1}(\sigma \cap \Delta_\gamma) \subset E_\gamma$. These ambiguities are given by periods of $E_\gamma$.

In the more general case, if $\gamma \subset \Gamma$ is a subgraph
with worse than logarithmic divergences, so that the pullback $\pi^*_\gamma(\omega_\Gamma)$
has a pole of order $k$ along $E_\gamma$, then the ambiguities in the contour
regularization of the Feynman amplitude are given by periods of the Hodge filtration 
of the primitive part of the cohomology, $F^{(d |\BV_\Gamma| -1 -k)} 
H^{d |\BV_\Gamma| -1}_{prim}(E_\gamma)$.
\end{Proposition}

\proof The Poincar\'e residue is dual to the Leray coboundary, in the sense that, if
$\omega$ is an $n$-form with logarithmic poles along a hypersurface $Y\subset X$, and
$\sigma$ is an $(n-1)$-chain in $Y$, then
$$ \frac{1}{2\pi i} \int_{\cL(\sigma)} \omega = \int_\sigma {\rm Res}[\omega]. $$
Thus, the ambiguity in the choice of a domain of integration $\Sigma_\epsilon$ as in
\eqref{regLeray}, which is up to the value of the integral 
$$ \int_{\cL_\epsilon(\tilde\sigma_\Gamma)} \pi^*(\omega_\Gamma) $$
is measured by (integral multiples of) \eqref{resambi}. 

In general, the form $\pi^*(\omega_\Gamma)$ has further
singularities on $\cL_\epsilon(\tilde\sigma_\Gamma)$. These come from
the intersections of $\cL_\epsilon(\tilde\sigma_\Gamma)$ with the preimage
$\tilde Z_\Gamma$ of the graph hypersurface of \eqref{eqn_g_hyp}. 

For sufficiently small $\epsilon >0$, up to a locus of codimension at least two,
the intersections $\tilde Z_\Gamma \cap \cL_\epsilon(\tilde\sigma_\Gamma)$
are coming from the components of $Z_\Gamma$ associated to those 
exceptional divisors $E_{\gamma'}$ 
that have non-empty intersection $E_\gamma \cap E_{\gamma'}\neq \emptyset$,
and such that $\gamma'$ is also a divergent subgraph of $\Gamma$.

In the model case where there would be only one divergent graph $\gamma$,
which is a logarithmic divergence, the form $\pi^*(\omega_\Gamma)$ would have no
further singularities on $\cL_\epsilon(\tilde\sigma_\Gamma)$ and 
the values of the integral \eqref{resambi} would then be periods 
$$ \begin{array}{ccccc}
H^{d|\BV_\Gamma|-1}(E_\gamma) & \times & H_{d|\BV_\Gamma|-1}(E_\gamma) & \to & \C \\[3mm]
{\rm Res}[\pi^*_\gamma(\omega_\Gamma)] &  & \tilde\sigma_\gamma & \mapsto &
\int_{\tilde\sigma_\gamma} {\rm Res}[\pi^*_\gamma(\omega_\Gamma)] .
\end{array} $$
In the more general case of a higher order pole, the resulting period pairing would be 
with the part of the cohomology that comes from ${\rm Res}(F^{d |\BV_\Gamma| -k} 
H^n( \overline{Conf}_\Gamma(X) \smallsetminus E_\gamma )$ which gives the 
piece of the primitive cohomology
$F^{(d |\BV_\Gamma| -1 -k)} H^{d |\BV_\Gamma| -1}_{prim}(E_\gamma)$, as in \cite{Griff},
\cite{Deligne}.

However, in general, there will be other divergent subgraphs $\gamma'$ with 
$E_{\gamma'}\cap E_\gamma \neq \emptyset$. In this case, assuming only
log divergences are present, one ends up with an iterated residue as in \eqref{iterRes},
with values in the cohomology of the intersection of all the corresponding
exceptional divisors.
\endproof

The integrals along the Leray coboundaries measure residues 
around the exceptional divisors $E_\gamma$ of the blowups, in a way
similar to what happens with the toric blowups of \cite{BK} for the
Feynman integrals in momentum space. The formulae described in the
previous sections for the motive of the wonderful compactification of
the configuration spaces show that, if the underlying smooth (quasi)projective
variety $X$ is mixed Tate as a motive, then the $E_\gamma$, their intersections, 
and the complements $\overline{Conf}_\Gamma(X) \smallsetminus E_\gamma$ that 
appear in the above are also mixed Tate, so that the ambiguities (the residues) in the
Leray regularization of the Feynman amplitudes are by periods of mixed Tate motives.

However, more generally, one considers the full integral  
$$ \int_{\overline{Conf}_\Gamma(M)} \pi^*(\omega_\Gamma) $$
and its regularization 
$$ \int_{\overline{Conf}_\Gamma(M)} \pi^*(\omega_\Gamma) 
-\sum_{\gamma\in \cG_\Gamma^{log}}\left( \int_{\overline{Conf}_\Gamma(M) \cap D_\epsilon(E_\gamma)} \pi^*(\omega_\Gamma) 
- \int_{\Sigma_\epsilon(\tilde\sigma_\gamma)} \pi^*(\omega_\Gamma) \right).
$$
In order to view these integrations as period computations, one needs to work
with the complement $\overline{Conf}_\Gamma(X)\smallsetminus Z_\Gamma$,
for which we do not have a comparably simple description of the motive. In particular,
the components of the graph hypersurface $Z_\Gamma$ are cones, which are 
simple to understand when one restricts them to a tubular neighborhood of one
of the divisors $E_\gamma$, as we have seen above. However, these cones
intersect in complicated ways outside of these tubular neighborhoods, so that
one does not have a good control over the motivic nature of these intersections.

\bigskip

{\bf Acknowledgments.} Part of this work was carried out during a visit of the
first author to the California Institute of Technology and during a visit of
both authors to the Max Planck Institute for Mathematics in Bonn. The
first author is partially  supported by a NWO grant; the second author is
partially supported by NSF grants DMS-0651925, DMS-0901221, and DMS-1007207.


\begin{thebibliography}{10}

\newcommand\arxiv[1]{
\href{http://arxiv.org/abs/#1}{\tt arXiv:#1}}




\bibitem{AluMa3} P.~Aluffi, M.~Marcolli, 
{\em Parametric Feynman integrals and determinant hypersurfaces},
arXiv:0901.2107, to appear in Advances in Theoretical and Mathematical Physics.

\bibitem{AluMa5} P.~Aluffi, M.~Marcolli, {\em Graph hypersurfaces and a dichotomy in the Grothendieck ring}, arXiv:1005.4470.

\bibitem{ATYu} L.A.~Aizenberg, A.K.~Tsikh, A.P.~Yuzhakov, {\em Multidimensional residues
and applications}, in ``Several Complex Variables, II", Encyclopedia of Mathematical
Sciences, Vol.8,  Springer Verlag, 1994.

\bibitem{AxSing1} S.~Axelrod, I.M.~Singer, {\em Chern-Simons perturbation theory}, in
``Proceedings of the XXth International Conference on Differential Geometric Methods in Theoretical Physics", Vol. 1, 2, pp.3--45, World Scientific, 1992.

\bibitem{AxSing2} S.~Axelrod, I.M.~Singer, {\em Chern-Simons perturbation theory. II.} 
J. Differential Geom. 39 (1994), no. 1, 173--213.

\bibitem{BelBr1} P. Belkale, P. Brosnan, {\it Matroids, motives, and a conjecture of Kontsevich.}
Duke Math. Journal, Vol.116 (2003) 147--188.

\bibitem{BelBr2} P. Belkale, P. Brosnan, {\it Periods and Igusa local zeta functions.} Int. Math. 
Res. Not. 2003, no. 49, 2655--2670.

\bibitem{BerKr} C.~Bergbauer, D.~Kreimer, {\em The Hopf algebra of rooted trees in
Epstein--Glaser renormalization}, Ann. Henri Poincar\'e, Vol.6 (2004) 343--367.

\bibitem{BerBruKr} C.~Bergbauer, R.~Brunetti, D.~Kreimer, {\em Renormalization
and resolution of singularities}, arXiv:0908.0633.

\bibitem{Bittner}  F.~Bittner, {\em The universal Euler characteristic for varieties of characteristic zero}. Compos. Math. 140 (2004), no. 4, 1011--1032.

\bibitem{BjDr2} J.~Bjorken, S.~Drell,  {\em Relativistic Quantum Fields}, 
McGraw-Hill, 1965.

\bibitem{Blo} S. Bloch, {\it Motives associated to sums of graphs.} preprint \arxiv{0810.1313}.

\bibitem{BEK} S.~Bloch, H.~Esnault, D.~Kreimer, {\em
On motives associated to graph polynomials}.  
Comm. Math. Phys.  267  (2006),  no. 1, 181--225.

\bibitem{BK} S.~Bloch, D.~Kreimer, {\em Mixed Hodge structures and
renormalization in physics}.  Commun. Number Theory Phys., Vol.2  (2008),
no. 4, 637--718. 

\bibitem{Bond} M.V.~Bondarko, {\em 
Differential graded motives: weight complex, weight filtrations and spectral sequences for realizations; Voevodsky versus Hanamura},  J. Inst. Math. Jussieu 8 (2009), no. 1, 39--97. 

\bibitem{BT} R.~Bott, C.~Taubes, {\em On the self-linking of knots. Topology and physics}. J. Math. Phys. 35 (1994), no. 10, 5247--5287.

\bibitem{BruKre} D.J.~Broadhurst, D.~Kreimer, {\it Association of multiple zeta values 
with positive knots via Feynman diagrams up to 9 loops.} Phys. Lett. B, 393(3-4):403Ð412, 
1997.

\bibitem{Brown} F.~Brown, {\em On the periods of some Feynman integrals},
arXiv:0910.0114.

\bibitem{BroSch} F.~Brown, O.~Schnetz, {\em A K3 in phi4}, arXiv:1006.4064.

\bibitem{Cey} \"O. Ceyhan, {\it An incarnation of Connes--Marcolli's renormalization group in 
Epstein--Glaser scheme}. submitted. \arxiv{}

\bibitem{Cheah} J.~Cheah, {\em The Hodge polynomial of the Fulton--MacPherson compactification of configuration spaces}, Amer. J. Math. 118 (1996), no. 5, 963--977. 

\bibitem{ck1} A.~Connes, D.~Kreimer, {\it Renormalization in quantum field theory and the 
Riemann--Hilbert  problem. I. The Hopf algebra structure of graphs and the main theorem.} 
Comm. Math. Phys. 210 (2000), no. 1, 249--273.

\bibitem{ck2} A.~Connes, D.~Kreimer, {\it Renormalization in quantum field theory and the 
Riemann--Hilbert problem. II. The $\beta$-function, diffeomorphisms and the renormalization 
group.}  Comm. Math. Phys. 216 (2001), no. 1, 215--241.

\bibitem{cm} A.~Connes, M.~Marcolli,  {\it From physics to number theory via noncommutative 
geometry II: Renormalization, the Riemann-Hilbert correspondence, and motivic Galois theory}. 
in "Frontiers in Number Theory, Physics, and Geometry, II" pp.617--713, Springer Verlag, 2006.

\bibitem{cm1} A.~Connes, M.~Marcolli,   {\it Noncommutative Geometry, Quantum Fields 
and Motives},Colloquium Publications, Vol.55, American Mathematical Society, 2008.

\bibitem{DecoPro} C.~De Concini, C.~Procesi, {\em Wonderful models of subspace
arrangements}, Selecta Math. (N.S.) Vol.1 (1995) N.3, 459--494.

\bibitem{Deligne} P.~Deligne, {\em \'Equations diff\'erentielles \`a points singuliers r\'eguliers}, Lecture Notes in Math., 163, Springer, Berlin, 1970.

\bibitem{DelDim} P.~Deligne, A.~Dimca, {\em
Filtrations de Hodge et par l'ordre du p\^{o}le pour les hypersurfaces singuli\`eres}. 
Ann. Sci. \'Ecole Norm. Sup. (4) 23 (1990), no. 4, 645--656. 

\bibitem{DimSa} A.~Dimca, M.~Saito, {\em A generalization of Griffiths's theorem on rational integrals}. Duke Math. J. 135 (2006), no. 2, 303--326. 

\bibitem{Doryn} D.~Doryn, {\em On one example and one counterexample in counting rational points on graph hypersurfaces}, arXiv:1006.3533.

\bibitem{EpGla} H.~Epstein, V.~Glaser, {\em The role of locality in perturbation theory},
 Ann. Inst. H.  Poincar\'e Sect. A (N.S.) 19 (1973), 211--295 (1974).

\bibitem{fm}  W.~Fulton, R.~MacPherson, {\it A compactification of configuration spaces.} 
Ann. of Math. (2) 139 (1994), no. 1, 183--225.

\bibitem{Getzler} E.~Getzler, {\em Mixed Hodge structures of configuration spaces},
alg-geom/9510018.


\bibitem{GilSou} H.~Gillet, C.~Soul\'e, {\em Descent, motives and $K$-theory}. J. Reine Angew. Math. 478 (1996), 127--176.

\bibitem{Griff} Ph.~Griffiths, {\em On the periods of certain rational integrals. I, II}, Ann. of Math. 
(2) 90 (1969),  460--495; 496--541.

\bibitem{GrHa} Ph.~Griffiths, J.~Harris, {\em Principles of algebraic geometry}, Wiley, 1994.

\bibitem{ItZu} C.~Itzykson, J.B.~Zuber, {\em Quantum Field Theory},
Dover Publications, 2006.

\bibitem{kon1} M.~Kontsevich, {\it Feynman diagrams and low-dimensional topology,} 
First European Congress of Mathematics, 1992, Paris, Volume II, Progress in Mathematics 
120, Birkhauser 1994, 97-121.



\bibitem{kon2} M.~Kontsevich, {\it Operads and motives in deformation quantization.} Lett. Math. 
Phys. 48 (1999), no. 1, 35--72.

\bibitem{kon3} M.~Kontsevich, {\it Deformation quantization of Poisson manifolds.} Lett. Math. 
Phys. 66 (2003), no. 3, 157--216.

\bibitem{kt} G.~Kuperberg, D.~Thurston, {\it Perturbative 3-manifold invariants by cut--and--paste 
topology.} preprint. \arxiv{math.GT/9912167}.

\bibitem{li} L.~Li, {\it Chow Motive of FultonÐMacPherson Configuration Spaces and Wonderful 
Compactifications.} Michigan Math. J. 58 (2009).

\bibitem{li2} L.~Li, {\em Wonderful compactification of an arrangement of subvarieties}. Michigan Math. J. 58 (2009), no. 2, 535--563.

\bibitem{mpro} R.~MacPherson, C.~Procesi, {\em Making conical compactifications
wonderful}, Selecta Math. (N.S.) 4 (1998) 125--139.

\bibitem{Manin} Yu.I.~Manin, {\em Generating functions in algebraic geometry and sums over trees}, in ``The moduli space of curves" (Texel Island, 1994), 401--417, 
Progr. Math., 129, Birkh\"auser, 1995. 

\bibitem{Man} Yu.I.~Manin, {\em Correspondences, motifs and monoidal transformations}. 
Mat. Sb. (N.S.) 77 (119) 1968 475--507. 

\bibitem{mar} M.~Marcolli, {Feynman motives}, World Scientific, 2010.

\bibitem{Mar} M.~Marcolli, {\em Motivic renormalization and singularities}, to appear 
in ``Quanta of Maths", Clay Institute Publications.
 
\bibitem{Nikolov} N.~Nikolov, {\em Cohomological analysis of the Epstein--Glaser renormalization}, arXiv:0712.2194.

\bibitem{Stem} J.~Stembridge, {\em Counting points on varieties over finite 
fields related to a conjecture of Kontsevich}, Ann. Combin. 2 (1998) 365--385.

\bibitem{Uly} A.P.~Ulyanov, {\em Polydiagonal compactification of configuration spaces},
J. Alg. Geom. 11 (2002) 129--159. 

\bibitem{voe} V.~Voevodsky, {\it Triangulated categories of motives over a field} in ÒCycles, 
transfer and motivic homology theories, pp. 188Ð238, Annals of Mathematical Studies, Vol. 
143, Princeton, 2000.

\bibitem{z} E.~Zeidler, {\it Quantum field theory. II. Quantum electrodynamics. A bridge between 
mathematicians and physicists.} Springer-Verlag, Berlin, 2009. xxxviii+1101 pp. 

\end{thebibliography}
\end{document}